\RequirePackage[fleqn]{amsmath}
\documentclass{llncs}

\usepackage[utf8]{inputenc}
\usepackage[T1]{fontenc}
\usepackage{lmodern}

\usepackage{microtype}
\usepackage[ngerman,english]{babel}
\usepackage{tikz}
\usepackage{amssymb}
\usepackage{xspace}
\usepackage{cite}

\usepackage[hidelinks]{hyperref}
\usepackage[nameinlink]{cleveref}

\newtheorem{observation}{Observation}[chapter]

\crefname{section}{Section}{Sections}
\crefname{definition}{Definition}{Definitions}
\crefname{figure}{Figure}{Figures}
\crefname{theorem}{Theorem}{Theorems}
\crefname{lemma}{Lemma}{Lemmata}
\crefname{corollary}{Corollary}{Corollary}
\crefname{observation}{Observation}{Observations}
\crefname{claim}{Claim}{Claims}

\usepackage{cancel}
\usepackage{subcaption}
\usepackage{algorithm}
\usepackage{algorithmicx}
\usepackage[noend]{algpseudocode}
\usepackage{mathtools,bm}
\algnewcommand{\LineComment}[1]{\State \(\triangleright\) #1}
\usetikzlibrary {arrows.meta,bending,positioning}
\usepackage{forest}
\tikzset{
	treenode/.style = {align=center, inner sep=0pt, text centered,
		font=\sffamily\Large\bfseries, text width=5cm},
	redded/.style = {treenode, circle, white, draw=red, fill=red, text width=1.5em, very thick,sibling distance=30mm},
	greened/.style = {treenode, circle, white, draw=green, fill=green, text width=1.5em, very thick,sibling distance=15mm},
	deleted/.style = {treenode, circle, white, draw=gray, fill=gray, text width=1.5em, very thick,sibling distance=15mm},
	unvisited/.style = {treenode, circle, black, draw=black, fill=white, text width=1.5em, sibling distance=7.5mm}
}

\newcommand{\N}{\ensuremath{\mathbb{N}}}
\newcommand{\naturalnumberpositive}{\ensuremath{{ \mathbb{N}^+ }}}

\newcommand{\maxk}{\ensuremath{\max^k}}
\newcommand{\alg}{\ensuremath{\textup{\textsc{Alg}}}\xspace}
\newcommand{\modi}{\,\,\text{mod}_1\,}
\newcommand{\textscalt}[1]{\textup{\textsc{#1}}}

\begin{document}
\title{Zero-Memory Graph Exploration with Unknown Inports}
\titlerunning{Zero-Memory Graph Exploration}
\author{Hans-Joachim Böckenhauer\inst{1}\orcidID{0000-0001-9164-3674} \and
	Fabian Frei\inst{1}\thanks{Part of the work by Fabian 
	Frei done during a stay at Hosei 
University, supported by grant GR20109 by the Swiss National 
Science Foundation (SNSF) and the Japan Society for the Promotion 
of Science (JSPS).}\orcidID{0000-0002-1368-3205} 
	\and
	Walter Unger\inst{2} \and
	David Wehner\inst{1}\orcidID{0000-0003-0201-4898}}
\authorrunning{Böckenhauer et al.}
\institute{ETH Zurich, Switzerland\\
	\email{\{hjb, fabian.frei, david.wehner\}@inf.ethz.ch} \and
	RWTH Aachen, Germany\\
	\email{quax@algo.rwth-aachen.de}}

\maketitle

	\begin{abstract}
		We study a very restrictive graph exploration problem.
		In our model, an agent without persistent memory is placed on a vertex of a graph and only sees the adjacent vertices. The goal is to visit every vertex of the graph, return to the start vertex, and terminate. The agent does not know through which edge it entered a vertex.
		The agent may color the current vertex and can see the colors of the neighboring vertices in an arbitrary order. The agent may not recolor a vertex.
		We investigate the number of colors necessary and sufficient to explore all graphs.
		We prove that $n-1$ colors are necessary and sufficient for exploration in general, $3$ colors are necessary and sufficient if only trees are to be explored, and $\min\{2k-3,n-1\}$ colors are necessary and $\min\{2k-1,n-1\}$ colors are sufficient on graphs of size $n$ and circumference $k$, where the circumference is the length of a longest cycle.
		This only holds if an algorithm has to explore all graphs and not merely certain graph classes.
		We give an example for a graph class where each graph can be explored with $4$ colors, although the graphs have maximal circumference.
		Moreover, we prove that recoloring vertices is very powerful by designing an algorithm with recoloring that uses only $7$ colors and explores all graphs.
	\end{abstract}
	
	\section{Introduction}\label{zm:sec:introduction}
	
	Say you wake up one morning in an unknown hotel with the desire 
	to stroll around and visit every place in the city. Considering 
	your terrible headache, you don't bother to remember anything 
	about which places you have visited already, but still, at the 
	end of the day, you want to return to your hotel. You know this 
	is not possible without further aid, so you decide to take some 
	crayons with you and color every place you visit. You are 
	endowed with keen eyes and you're able to see the colors of the 
	places around you. All you now need to know is how many colors 
	you have to take along and how you color the places. This paper 
	deals exactly with that situation.
	
	The exploration of an unknown environment by a mobile entity is one of the basic tasks in many areas.
	Its applications range from robot navigation over image recognition to sending messages over a network.
	Due to the manifold purposes, there is a great deal of different settings in which exploration has been analyzed.
	In this paper, we consider the fundamental problem of a single 
	agent, e.g., a robot or a software agent, that has to first 
	explore all vertices of an initially unknown undirected graph 
	$G$, then return to the start vertex and terminate.
	By exploring we mean that the agent is located at a vertex and can, in each step, either go to an adjacent vertex or terminate.
	
	If the vertices of $G$ have unique labels and without further 
	restrictions on the agent, this becomes a trivial task.
	However, in many applications, the environment is unknown and the agent is a simple and inexpensive device.
	Hence, we consider \emph{anonymous} graphs, that is, there are no unique labels on the vertices or edges.
	Moreover, the edges have no port labels;
	the labeling is given implicitly by the order in which the agent sees the edges and can be different at each visit of a vertex.
	The agent itself is \emph{oblivious}, that is, it has no persistent memory.
	Such agents are sometimes also called \emph{zero-memory algorithms} or \emph{$1$-state robots} \cite{Cohen2008, Das2019}.
	
	Clearly, with these restrictions, there does not exist a feasible exploration algorithm. In fact, not even a graph consisting of one single edge could be explored since the algorithm would not know when to terminate.
	Therefore, in most models with anonymous graphs and oblivious agents, the agent remembers through which port it entered a vertex.
	We, in contrast, assume that the agent does not know through which edge it entered a vertex.
	This is sometimes called \emph{unknown inports} or \emph{no inports} \cite{Menc2017}.
	Instead, we allow the agent to color the current vertex, a feature which is also referred to as placing distinguishable pebbles \cite{Disser2019} or labeling vertices \cite{Cohen2008}.
	In this paper, we prefer the notion of coloring\footnote{Note 
	that this coloring is just a normal labeling and has nothing to 
	do with graph coloring such as in 3-\textscalt{Coloring}; it is 
	perfectly fine to color adjacent vertices with the same color.}.
	This notion emphasizes that a colored vertex may never be recolored unless we explicitly allow it and then use the term ``recoloring.'' 
	However, having the ability to color vertices alone is still 
	utterly useless for an oblivious agent that has to return to 
	the start vertex to terminate.
	Therefore, we relax our restrictions by allowing the agent to see the labels (i.e., colors) of the neighboring vertices.
	
	We consider storage efficiency and analyze the minimum amount of colors necessary and sufficient to explore any graph with $n$ vertices.
	We prove that 3 colors are both necessary and sufficient to explore trees, whereas $n-1$ colors are necessary and sufficient to explore every graph with $n$ vertices. This striking difference is not limited to planarity, graphs of large treewidth or graphs with a large feedback vertex set; in fact, even planar graphs with treewidth 2 and feedback vertex set number 1 need $\Omega(n)$ colors. We discover that the driving parameter of a graph is its circumference, the length of a longest cycle. We show that $2k-1$ colors are sufficient and $\min\{2k-3, n-1\}$ colors are necessary to explore all graphs of circumference at most $k$.
	We further show that it is possible to explore all so-called squares of paths with 4 colors, which shows that, if an algorithm does not have to explore all possible graphs but only graphs of a particular graph class, much fewer colors suffice.
	Finally, we make an ostensibly inconspicuous change to our model and analyze the case where we allow recoloring vertices. We show that, in this model, 7 colors are enough to explore all graphs.
	
	This paper is organized as follows. In the rest of this introduction, we consider related work and lay out the basic definitions. In \cref{sec:trees}, we present the analysis when the graphs to be explored are trees. In \cref{sec:generalgraphs}, we analyze the general case before we then turn to graphs of a certain circumference in \cref{sec:circumference}. Afterwards, we consider squares of paths in \cref{sec:P2} and finally recoloring in \cref{sec:recoloring} before we conclude in \cref{sec:conclusion_zeromemory}.
	
	\subsection{Related work}
	There is a vast body of literature on exploration and 
	navigation problems.
	A great deal of aspects have been analyzed;
	in general, they can be categorized into one of the following four dimensions: environment, agent, goal, complexity measure.
	We briefly discuss these dimensions and highlight the setting considered in this paper.
	
	The environment dimension is concerned with whether there is a specific geometric setting or a more \emph{abstract setting} such as a graph or a \emph{graph class}. Sometimes there are special environmental features such as faulty links or, as in our case, \emph{local memory}, sometimes also called storage.
	In the agent dimension, we find aspects such as whether there 
	is a \emph{single agent} or whether there are multiple agents; 
	whether the agents are \emph{deterministic} or probabilistic; 
	whether the agents have some restrictions such as \emph{limited 
	memory}, range or \emph{view}; and whether the agents possess 
	special abilities such as \emph{marking of vertices} or 
	teleportation.
	Some goals include mapping the graph, finding a treasure, meeting, exploring all edges, and \emph{exploring all vertices}.
	For the latter, most researchers consider one of the following three modes of termination: perpetual exploration, where the agent has to visit every vertex infinitely often;
	exploration with stop, where the agent has to stop at some point after it has explored everything; 
	and \emph{exploration with return}, where the agent has to 
	return to the starting point after the exploration and then 
	terminate.
	The main complexity measures are time complexity, space/memory complexity, \emph{storage complexity} and competitive ratio\footnote{Here, we use ``time complexity'' as the complexity of the algorithm that calculates the decisions of the agent and ``competitive ratio'' as the number of time steps of the agent compared to an optimal number of time steps.}.
	
	As highlighted above, in this paper, we focus on vertex 
	exploration by a single agent with local memory or storage. For 
	an overview on this segment of exploration problems, we refer 
	to the excellent overview of Das~\cite{Das2019} and the first 
	two chapters of~\cite{Gasieniec2008} by G\k{a}sieniec and 
	Radzik. 
	We are not aware of any research on the exploration of anonymous graphs by oblivious agents where the labels of the neighboring vertices are visible. However, the models of Cohen et al.~\cite{Cohen2008} and Disser et al.~\cite{Disser2019} are similar to ours; they analyze graph exploration with anonymous graphs, local port labels, a single oblivious agent, and the ability to label the current vertex.
	
	In the model of Cohen et al.~\cite{Cohen2008}, there is a robot $\mathcal{R}$ with a finite number of states that has to explore all vertices of an unknown undirected graph and then terminate. The robot sees at each vertex the incident edges as \emph{port numbers}. The order of these edges is fixed per vertex, but unknown to $\mathcal{R}$. The robot knows through which port it entered a vertex. In a preprocessing state, the vertices are labeled with pairwise different labels. Cohen et al.~analyzed how many labels are necessary to explore all graphs. They proved that a robot with constant memory can explore all graphs with just three labels. Moreover, they showed that for any $d>4$, an oblivious robot that uses at most $\lfloor \log d \rfloor-2$ pairwise different labels cannot explore all graphs of maximum degree $d$.
	
	Their model does not directly compare to ours: Our labels are not assigned in a preprocessing stage, but during the exploration by the algorithm/robot. Moreover, the incoming port number is not known, the order of the port numbering is not fixed, the labels may not be changed, the algorithm has to return to the start vertex, and, most importantly, the algorithm sees the labels of the adjacent vertices. However, even though the models are quite different, we see that known inports is a much stronger feature than seeing the labels of neighboring vertices: Our general lower bound does not depend on the maximum degree, but on the number of vertices in the graph. As we will see later, we provide a lower bound that is linear in the number of vertices even when the maximum degree is restricted to $3$.
	
	The model of Disser~et~al.~\cite{Disser2019} is even closer to our model. Here, the labels---they call them distinguishable pebbles---are assigned during the exploration by the agent/algorithm as well. Moreover, the goal is exploration with return. They showed that, for any agent with sub-logarithmic memory, $\Theta(\log\log n)$ labels are necessary and sufficient to explore any graph with $n$ vertices. Moreover, they characterized the trade-off between memory and the number of labels as follows. When the agent has $\Omega(\log(n))$ bits of memory, all graphs on $n$ vertices can be explored without any labels. As soon as the agent only has $\mathcal{O}(\log(n)^{1-\varepsilon})$ bits of memory, $\Omega(\log\log(n))$ labels are needed to explore all graphs on $n$ vertices. However, with that many labels, even a constant amount of memory suffices for the exploration.
	
	As before, this model does not directly compare to ours since neither is contained in the other. Their results seem to support the idea that knowing inports is stronger than seeing the labels of neighboring vertices; however, in contrast to our model, the focus of their work was on constant or sub-logarithmic memory.

	\subsection{Basic Definitions}
	We use the usual notions from graph theory as found for example in the textbook by Diestel~\cite{Diestel2017}. 
	The graph exploration setting considered in this paper is defined as follows.
	We first describe the setting informally. An agent is placed on a vertex, called the start vertex, of an undirected connected graph and moves along edges, one edge per step. 
	In a step, the agent may use an arbitrary natural number to 
	color the vertex on which it is currently located if this 
	vertex was uncolored up to now, and then move to a neighbor. 
	As basis for its decisions, the agent may only use the color of the current vertex and the colors of the neighboring vertices. 
	The agent can neither use the identity of the vertices nor any numbering of the edges. 
	The agent has no persistent memory; in particular, the agent does not know through which edge it entered the current vertex (if any; the start vertex is not entered from anywhere).
	On the basis of the coloring of the current vertex and the neighbors, the agent chooses a neighbor it wants to move to, or it decides to stop. For no decision or choice can the agent distinguish between neighbors that have the same color.
	
	The task is to provide a strategy for the agent---which is called an algorithm \alg---that determines for each situation the action of the agent in such a way that no matter on which graph and on which start vertex the agent is placed, and no matter which neighbors are used to go to if there is a choice, the agent will visit all vertices, return to the start vertex, and stop there.
	
	As with classical online problems, the concept of an \emph{adversary} is useful to formulate bounds on the number of colors necessary and sufficient to explore all graphs. 
	This adversary makes the decisions that are left open in the informal description above, namely choosing the start vertex and choosing the vertex the agent visits next if there is a choice among neighbors of the same color to which the agent wants to go.
	For all graphs and for all possible choices of the adversary, the agent must visit all vertices and then stop on the start vertex.
	
	An algorithm in this model is a function that determines what the agent should do when located on a vertex with a certain color structure in the neighborhood. We denote the color structure by a pair $(c_0, E)$, where $c_0 \in \N$ stands for the color of the current vertex and $E\colon \N \to \N$, where $E$ is $0$ almost everywhere, stands for the colors of the neighbors. For a number $c \in \N$, $E(c)$ stands for the number of neighbors of color $c$. In cases where we consider only a restricted amount of colors, we consider functions $c_0$ and $E$ of smaller domain, for example colors $c_0 \in \{1,\ldots,n_c\}$, and $E\colon \{1,\ldots,n_c\} \to \N$. We denote the set of all pairs $(c_0, E)$ by $\mathcal{E}$ and call it the \emph{environment}.
	
	An \emph{algorithm} in this model is a function
	\[\text{move}\colon \mathcal{E} \to \N \times \N \cup 
	\{\textscalt{Stop}\},\]
	with the restriction that whenever $\text{move}(c_0,E) = 
	(c_1,d)$, we must have $E(d)>0$, that is, the algorithm is not 
	allowed to send the agent to a neighbor that does not exist. 
	When the number of colors allowed is at most $n_c$, the range 
	of $\text{move}$ is $\{1,\ldots,n_c\} \times \{1,\ldots,n_c\} 
	\times \{\textscalt{Stop}\}$.
	Moreover, we must have $c_1 = c_0$ unless $c_0 = 0$, that is, the agent may only color the current vertex if it was not colored before. We analyze the model where this last restriction is canceled in \Cref{sec:recoloring}.
	
	To carry out an algorithm on a given graph $G$, we use an adversary, and carry out steps. Initially, all vertices are uncolored, that is, have color $0$.
	The adversary selects a start vertex $v_0$ and places the agent there. 
	Then, each step works as follows: The agent is located in a vertex $v$. 
	The coloring in the neighborhood is translated into an element $(c_0, f)$ of $\mathcal{E}$ in the obvious way, by counting the number of neighbors for each occurring color. 
	The value $\text{move}(c_0,f)$ determined by the algorithm is then used as follows. 
	If it is \textscalt{Stop}, the run of the algorithm stops.
	If it is a pair $(c_1,d)$, vertex $v$ is colored with color $c_1$, and the adversary chooses a neighbor of $v$ of color $d$, to which the agent moves for the next step.
	
	An algorithm \emph{successfully explores all graphs} if for all 
	graphs $G$ and for all adversaries, after finitely many steps, 
	all vertices of $G$ have been visited (they do not need to have 
	a color $c>0$, but they must have been the ``current vertex'' 
	in some step), the agent is located on the start vertex $v_0$, 
	and the decision of the algorithm is \textscalt{Stop}.
	An algorithm is correct on a graph class $\mathcal{C}$ if it 
	successfully explores all graphs $G \in \mathcal{C}$.
	
	Throughout the paper, we denote by $n$ the number of vertices of $G$ and use $[n]$ to denote $\{1,\ldots,n\}$. For a vertex $v \in V$, we write $c(v) \in \N$ for its color; $N(v)$ for the open neighborhood of $v$, that is, the set of vertices adjacent to $v$; and $N[v]$ for the closed neighborhood of $v$, where $N[v]\coloneqq N(v)\cup\{v\}$. We use $\text{mod}_1$ to denote a modulo operator shifted by $1$, i.e.,
	\[n \,\,\text{mod}_1 \,\,m \coloneqq ((n-1) \mod m) + 1.\]
	Instead of having the numbers $0,\ldots,m-1$ as the outcome of the modulo operation, one thus obtains the numbers $1,\ldots,m$.
	For convenience, we sometimes speak of an algorithm behaving in 
	a particular way and mean by this formulation that the agent of 
	the algorithm behaves in a particular way.

	\section{Exploration of Trees}\label{sec:trees}
	We begin by analyzing the problem on trees.
	We show that only three colors are enough to explore all trees.
	In line with the research that analyzes graph exploration with pebbles, we do not count $0$ as a color.
	The idea of the algorithm is to alternately label the vertices with colors $1$, $2$, and $3$ and then follow a simple depth-first strategy.
	Since there are no cycles, there is always a unique path back to the start vertex and backtracking is possible.
	Moreover, the start vertex is recognized by the fact that it is the only vertex with color $1$ where all neighbors have color $2$.
	This strategy is formalized in the algorithm 
	\textscalt{TreeExploration} below.
	
	\begin{algorithm}[htbp]
		\caption{{}\textscalt{TreeExploration}} \label{alg:tree}
		\textbf{Input:} An undirected tree in an online fashion. In each step, the input is $c(v) \in \N$, which is the color of the current vertex $v$, and $c(v_1), \ldots, c(v_d)$, the colors of the neighbors of $v$, where $d$ is the degree of $v$.\\
		\textbf{Output:} In each step, the algorithm outputs $c(v)$, the color the agent assigns to $v$, and $i \in [d]$, the vertex the agent goes to next.\\
		\textbf{Description:} Assign alternately color $1$, $2$, and $3$ in order to achieve a sense of direction.
		\begin{algorithmic}[1] \label{alg:treeexploration}
			\If{$c(v)=0$}
			\State $c(v)\coloneqq \left(\left(\max\limits_{v' \in N(v)} c(v')\right)+1\right) \modi 3$ \label{line:treeplacepebbles}
			\EndIf
			\If{there is an uncolored neighbor} \label{line:tree_condition_unvisited_neighbor}
			\State go to some uncolored neighbor. \label{line:treeforward}
			\ElsIf{$c(v)=1$ and there is no neighbor $v'$ having $c(v')=3$}
			\State \textbf{terminate}. \label{line:treeterminate}
			\Else
			\State go to a neighbor $v'$ with 
			$c(v')=\left(c(v)-1\right) \,\,\text{mod}_1 \,\,3$. 
			\Comment{this neighbor is going to be unique} 
			\label{line:treegoto}
			\EndIf
		\end{algorithmic}
	\end{algorithm}
	
	Obviously, \textscalt{TreeExploration} uses at most $3$ colors.
	We prove first that the algorithm terminates on the start vertex and only on the start vertex.
	
	\begin{lemma}\label{lem:treeterminatesonlyuponreachingstartvertex}
		The agent terminates only upon reaching the start vertex.
	\end{lemma}
	\begin{proof}	
		The termination rule can only be applied on a vertex with color $1$ without any neighbor $v'$ with $c(v')=3$. Clearly, this is true for the start vertex. Apart from the start vertex, all vertices assigned color $1$ have a neighbor with color $3$ since there is no vertex with only uncolored neighbors except the start vertex.
	\end{proof}
	
	Then, we prove that \textscalt{TreeExploration} does indeed 
	visit every vertex and return to the start vertex.
	\begin{lemma}\label{lem:treeverticesareexplored}
		\textscalt{TreeExploration} explores all vertices and then 
		returns to the start vertex.
	\end{lemma}
	\begin{proof}
		The color assignment in line~\ref{line:treeplacepebbles} 
		has the effect that vertices are colored with $(d(v)+1) 
		\modi 3$, where $d(v)$ is the distance from the start 
		vertex. Clearly, since there are no cycles in the graph, 
		\textscalt{TreeExploration} always returns to the start 
		vertex via line~\ref{line:treegoto}.
		Once all neighbors of the start vertex are explored, 
		\textscalt{TreeExploration} terminates by 
		\cref{lem:treeterminatesonlyuponreachingstartvertex}.
		
		We now show that all vertices are explored before 
		\textscalt{TreeExploration} terminates:
		Assume towards contradiction that there is a tree $T = (V,E)$ where not all vertices are explored.
		Let $C \subset V$ be the vertices that are explored and thus colored and let $U \subset V$ be the vertices that are not explored.
		
		Since $T$ is connected, there is an edge $\{v,w\}$ between two vertices $v \in C$ and $w \in U$.
		The vertex $v$ is visited at some point.
		There is an unvisited neighbor, namely $w$, hence 
		line~\ref{line:tree_condition_unvisited_neighbor} is 
		applied and \textscalt{TreeExploration} visits an 
		unexplored neighbor, say $u \in C$.
		Let us call the set of vertices in $C$ whose shortest path to the start vertex contains $v$ the \emph{branch} of $v$. Clearly, $u$ is in the branch of $v$.
		
		After $u$, \textscalt{TreeExploration} might visit other 
		vertices in the branch of $v$.
		At some point, all vertices in the branch of $v$ are visited and only lines~\ref{line:treeterminate} or~\ref{line:treegoto} can be applied. By \cref{lem:treeterminatesonlyuponreachingstartvertex}, line~\ref{line:treeterminate} can only be applied on the start vertex and the start vertex is not in the branch of $v$.
		Therefore, only line~\ref{line:treegoto} can be applied.
		Since there are no cycles, \textscalt{TreeExploration} thus 
		returns to $v$ at some point.
		Now, however, all neighbors of $v$ except $w$ and maybe other neighbors in $U$ are colored. Therefore, $v$ has to visit $w$ or another neighbor in $U$ according to line~\ref{line:forward_condition}. This contradicts our assumption that $w \in U$.
		Hence, \textscalt{TreeExploration} explores all vertices.
	\end{proof}

	Together, these two lemmas yield the following theorem.
	\begin{theorem}
		\textscalt{TreeExploration} is correct on trees and uses at 
		most $3$ colors.
	\end{theorem}
	
	We prove that \textscalt{TreeExploration} is optimal in the 
	sense that it is impossible to use fewer than $3$ colors to 
	solve graph exploration in our model on trees. Before doing so, 
	we make the following crucial observation.
	
	\begin{observation}[Functional Nature]\label{obs:functionalnature}
		Due to its functional nature, once an algorithm is in a vertex $v$ where all vertices in $N[v]$ have been colored and takes a decision upon which it goes to a neighbor $w$ or upon which the adversary chooses neighbor $w$ as next vertex, this choice can be made by the adversary each time the algorithm returns to $v$.
		Moreover, once the algorithm is in a vertex $v$ and takes a decision to go to a neighbor of color $d$, the same decision will be made in a vertex $w$ if $c(w)=c(v)$ and the colors in $N(w)$ can be mapped bijectively to the colors in $N(v)$.
	\end{observation}
	
	\begin{theorem}\label{thm:trees_lower_bound}
		There is no algorithm that solves graph exploration as in our model on every tree and that uses less than 3 colors.
	\end{theorem}
	\begin{proof}
		Assume by contradiction that there exists such an algorithm.
		Consider a path with seven vertices.
		Denote them by $v_1$ to $v_7$ in sequential order. Let $v_1$ be the start vertex.
		
		For $i \in [5]$, if an agent does not color $v_i$, it cannot visit $v_{i+2}$:
		The agent could continue to $v_{i+1}$, but then it cannot distinguish between $v_i$ and $v_{i+2}$.
		Therefore, the adversary can make the agent go back to $v_i$. If the agent now does not color $v_i$\footnote{Note that the agent may have colored $v_{i+1}$ and thus the environment of the second visit of $v_i$ may be different from the environment of the first visit, leading to a potentially different decision.}, it is caught in an endless loop by \cref{obs:functionalnature}, which contradicts the assumption that the algorithm works correctly on all trees.
		
		When the agent has reached $v_7$, it has to return to $v_1$. It must make a step from $v_4$ to $v_3$ at some point. If $c(v_3) = c(v_5)$, the adversary can make the agent go back to $v_5$ instead, which again opens the way to an endless loop. Hence, $c(v_3)\neq c(v_5)$. The same is true at $v_3$ and at $v_2$. Hence, $c(v_4)\neq c(v_2)$ and $c(v_1)\neq c(v_3)$. Without loss of generality, assume $c(v_1)=1$. Then this implies that either $(c(v_1),\ldots,c(v_5))=(1,1,2,2,1)$ or $(c(v_1),\ldots,c(v_5))=(1,2,2,1,1)$.
		
		If $(c(v_1),\ldots,c(v_5))=(1,1,2,2,1)$, the agent has to go from $v_4$ to $v_3$; however, then the agent goes from $v_3$ back to $v_4$ because of its functional nature. Similarly, if $(c(v_1),\ldots,c(v_5))=(1,2,2,1,1)$, the agent has to go from $v_3$ to $v_2$, but then it goes from $v_2$ back to $v_3$.
	\end{proof}
	
	As an aside, we would like to note that the example with the path in the proof above does not work if we change our model and allow recoloring vertices.
	Then, it would be possible to successfully explore all paths with the start vertex being a leaf by just using one color in addition to the ``non-color'' $0$. A possible algorithm might proceed as follows.
	
	If $c(v)=0$ and there is an uncolored neighbor, set $c(v)=1$ and go to that neighbor; if $c(v)=1$ or if there is no uncolored neighbor, set $c(v)=0$ and go to the colored neighbor. If $c(v)=1$ and there is no colored neighbor, terminate.
	This way, the algorithm visits all vertices from the start leaf up to the other leaf and assigns to each vertex except the last one color $1$. From the last vertex, the algorithm goes back to the start vertex by always deleting the assigned color. This is possible since the color of the current vertex is given to the algorithm and the algorithm thus knows whether the vertex has been visited at some point before.
	
	We return to recoloring in more depth in \cref{sec:recoloring}, where we corroborate our finding that recoloring radically changes our model, resulting in much more powerful algorithms.
	
	\section{Exploration of General Graphs}\label{sec:generalgraphs}
	In this section, we provide bounds for the exploration of general graphs, that is, without restricting the graph class. We begin with an upper bound, then turn to a corresponding lower bound, and then diverge slightly and discuss non-uniform algorithms, where the algorithm knows the size of the input graph.
	\subsection{Upper Bound}\label{sec:generalgraphs_upperbound}
	How can an agent proceed on graphs that are not necessarily trees?
	Again, a depth-first search strategy suffices; however, this time, the agent uses almost as many colors as vertices.
	
	Let us first describe the idea.
	The agent starts at the start vertex and colors it with color $1$. This vertex is the root. All other vertices will receive larger colors; hence, the root is easy to recognize.
	We wish to carry out DFS. It is easy to go whenever possible to an unexplored---or, equivalently, uncolored---neighbor of the current vertex.
	But how shall the agent find the way back?
	The idea is to color a new vertex with a color $1$ larger than the largest color in the neighborhood.
	This ensures that colors increase along paths taken forward in the tree.
	A notable exception is when the agent arrives at a leaf in the DFS tree, that is, a vertex where all neighbors are colored.
	Then, the agent can save a color by assigning the largest color in the neighborhood---which is the color of the vertex the agent came from--instead of the largest color in the neighborhood plus one.
	For backtracking from a vertex $v$, the agent goes to a vertex whose color is one less than $c(v)$. Such a vertex always exists, except if $v$ is the root.
	The only thing one has to prove is that this neighbor is unique.
	The formal description of \textscalt{DepthFirstSearch} is 
	provided by Algorithm~\ref{alg:dfs}.
	
	Clearly, \textscalt{DepthFirstSearch} is well defined.
	We prove in the following that \textscalt{DepthFirstSearch} 
	explores all vertices and uses at most $n-1$ colors.
	We start by proving that whenever \textscalt{DepthFirstSearch} 
	goes to a vertex that has been colored before, this vertex is 
	the predecessor of the current vertex.
	By \emph{predecessor} of a vertex $v$ that is not the root, we mean the neighbor $w$ from which the agent moved to $v$ when $v$ was first visited.
	
	\begin{algorithm}
		\caption{{}\textscalt{DepthFirstSearch}}
		\label{alg:dfs}
		\begin{algorithmic}[1]
			\LineComment{Save one color by not assigning a new color for leaves or any vertices where all neighbors are colored}
			\If{$c(v)=0$ and all neighbors are colored}
			\State $c(v)\coloneqq \left(\max\limits_{v' \in N(v)}c(v')\right)$
			\State go to a neighbor $v'$ with maximal color \label{line:savecolors}
			\LineComment{Normal case: on a new vertex, there is another uncolored neighbor}
			\ElsIf{$c(v)=0$}
			\State $c(v)\coloneqq \left(\max\limits_{v' \in N(v)}c(v')\right)+1$ \label{line:placepebbles}
			\EndIf
			\If{there is a neighbor $v'$ with $c(v')=0$} \label{line:forward_condition}
			\State go to $v'$. If there are multiple such $v'$, go to the one that is presented first \label{line:forward}
			\ElsIf{there is a neighbor $v'$ with $c(v')=c(v)-1$}
			\State go to $v'$ \label{line:backwards}
			\Else
			\State \textbf{terminate.}
			\EndIf
		\end{algorithmic}
	\end{algorithm}
	
	\begin{lemma} \label{lem:directpredecessors}
		On any vertex $v$, either \textscalt{DepthFirstSearch} goes 
		to an unvisited vertex or it goes back to the predecessor 
		of $v$. In other words, whenever line~\ref{line:backwards} 
		is applied, $v'$ is the predecessor of $v$.
	\end{lemma}
	\begin{proof}
		We prove this by induction on the number of times line~\ref{line:backwards} is applied. 
		If line~\ref{line:backwards} is applied for the first time, there can only be one vertex $v'$ with color $c(v')=c(v)-1$ since assigning a color multiple times requires backtracking, i.e., line~\ref{line:backwards}.
		Moreover, by the same reasoning, $v$ indeed receives color $c(v')+1$; there cannot be another neighbor $z$ with color $c(z)>c(v')$.
		
		Assume that line~\ref{line:backwards} has been applied $k-1$ times and always the neighbor visited right before has been visited. 
		We show that for the $k$-th time when line~\ref{line:backwards} is applied, again the neighbor visited right before will be visited.
		Call the current vertex $v$ and let $v'$ be the vertex visited right before $v$. We prove two things. First, we show that there is only one neighbor $v''$ with color $c(v'')=c(v)-1$; second, we show that $c(v)=c(v')+1$.
		
		Assume there is another neighbor $z$ of $v$ with $c(z)>c(v')$. We distinguish two cases: Either $v'$ was visited before $z$ or vice versa. In the first case, by the induction hypothesis, the only way for the agent to go from $z$ back to $v'$ is via backtracking at some point from $z$ to its predecessor. However, this is not possible since line~\ref{line:backwards} can only be applied when there is no unvisited neighbor, but $v$ is such an unvisited neighbor. In the second case, consider the search sequence $z = v_1, v_2, \ldots, v_k = v'$ of the algorithm between $z$ and $v'$. The only way that $z$ is assigned a larger color than $v'$ is that line~\ref{line:backwards} is applied between $z$ and $v'$. By the induction hypothesis, whenever this happens, the direct predecessor is visited. So either line~\ref{line:backwards} is applied right from $z=v_1$ or the algorithm goes to a previously unvisited vertex $v_2$ and $z$ occurs somewhere else in the sequence. At some point, the agent has to backtrack from $z$, which is not possible since $z$ has an unvisited neighbor. Therefore, such a neighbor $z$ does not exist. 
		
		Assume there are two vertices $v'$ and $v''$ with $c(v')=c(v'')=c(v)-1$ and the agent has visited $v''$ some time before $v'$. Let $v''=v_1, \ldots, v_s=v'$ be the sequence the algorithm took to get from $v''$ to $v'$.
		According to the induction hypothesis, whenever line~\ref{line:backwards} was applied in this sequence, a direct predecessor was visited.
		The vertex $v''$ might appear several times in this list.
		Let $v_t$ be the last vertex in the sequence that is equal to $v''$.
		By construction, colors change only by $1$ between two neighbors in this sequence. 
		If $c(v_{t+1})=c(v'')+1$, there would have to be another vertex $\bar{v}$ in the list with color $c(\bar{v})=c(v'')$.
		However, since there are only direct predecessors when applying line~\ref{line:backwards} and since the color cannot decrease otherwise, this is not possible.
		Hence, $c(v_{t+1})=c(v'')-1$. However, this is not possible either since line~\ref{line:backwards} cannot be applied when there is an unvisited neighbor, which is the case here with $v$.
		Therefore, the assumption that there are two vertices $v'$ and $v''$ with $c(v')=c(v'')=c(v)-1$ is wrong and there is only one such vertex. This is the direct predecessor of $v$, which proves the lemma.
	\end{proof}
	\begin{lemma}\label{lem:dfsexploresallvertices}
		\textscalt{DepthFirstSearch} explores all vertices.
	\end{lemma}
	\begin{proof}
		We can argue similar as for \cref{lem:treeverticesareexplored}.
		Assume towards contradiction that there is a graph $G = (V,E)$ where not all vertices are explored.
		Let $C \subset V$ be the vertices that are explored and thus colored and let $U \subset V$ be the vertices that are not explored.
		
		Since $G$ is connected, there is an edge $\{v,w\}$ between a vertex $v \in C$ and $w \in U$.
		The vertex $v$ is visited at some point. According to our 
		observation above, \textscalt{DepthFirstSearch} then visits 
		an unexplored neighbor.
		Going back to the predecessor of $v$ is not possible since $v$ has an unexplored neighbor, namely $w$, and line~\ref{line:backwards} is only evaluated if the condition of line~\ref{line:forward_condition} is wrong, i.e., if there is no unexplored neighbor.
		\textscalt{DepthFirstSearch} then continues to visit other 
		vertices in $C$. However, at some point, there are no 
		unvisited vertices in $C$ left and 
		\textscalt{DepthFirstSearch} has to return according to 
		line~\ref{line:backwards}.
		
		By \cref{lem:directpredecessors}, 
		\textscalt{DepthFirstSearch} returns to the direct 
		predecessor every time.
		Therefore, \textscalt{DepthFirstSearch} visits $v$ again.
		Now, \textscalt{DepthFirstSearch} may visit another 
		unexplored neighbor $s \in C$. 
		By the same arguments as before, however, 
		\textscalt{DepthFirstSearch} returns to $v$. 
		At some point, all neighbors of $v$ in $C$ are explored and the algorithm returns to $v$. 
		Then, the vertex $w$ has to be explored next according to line~\ref{line:forward_condition}.
		This contradicts our assumption that $w \in U$.
		Therefore, on every graph $G$, \textscalt{DepthFirstSearch} 
		explores all vertices.
	\end{proof}
	\begin{lemma}
		The agent in \textscalt{DepthFirstSearch} returns to the 
		start vertex and terminates.
	\end{lemma}
	\begin{proof}	
		The agent cannot get stuck since, for every vertex $v$ except the start vertex and except leaves, there is always a neighbor $v'$ with color $c(v')=c(v)-1$.
		Once every vertex is explored, the start vertex has no such neighbor, hence the algorithm terminates in such a case in the start vertex.
	\end{proof}
	Together with the obvious fact that 
	\textscalt{DepthFirstSearch} uses at most $n-1$ pairwise 
	different colors since no vertex is colored twice and since the 
	last vertex that is colored has no uncolored neighbors, these 
	lemmas prove the correctness of \textscalt{DepthFirstSearch}:
	\begin{theorem}
		\textscalt{DepthFirstSearch} is correct and uses at most 
		$n-1$ colors.
	\end{theorem}
	
	\subsection{Lower Bound}\label{sec:generalgraphs_lowerbound}
	At first glance and having seen that 
	\textscalt{TreeExploration} only uses three colors, we are 
	tempted to think that \textscalt{DepthFirstSearch} uses 
	unnecessarily many colors.
	However, we are going to show in the following that 
	\textscalt{DepthFirstSearch} is optimal in our model.
	Hence, our perspective changes.
	We now deal with an unknown algorithm of which we only know that it successfully explores all graphs.
	In order to show a lower bound, we have to define a graph and the decisions of an adversary such that every algorithm with a limited amount of colors fails to explore our graph given the decisions of our adversary.
	
	We start with an easy observation and two technical lemmas.
	These lemmas will be crucial in creating lower bounds for our model.
	They will allow us to define an instance and then argue that the agent has to take a certain path in this instance and cannot choose a different route.
	\begin{observation}\label{obs:cantignoreuncolored}
		If the agent of an algorithm \alg that successfully explores all graphs arrives at a colored vertex $v$ and there is an uncolored neighbor, then the agent has to visit an uncolored neighbor next.	
	\end{observation}
	\begin{proof}
		The uncolored neighbors might all be leaves in the input graph and if the agent decides to go to a colored neighbor, then it will always take this decision when located on $v$ and never explore the uncolored leaves, which could only be explored via a move from $v$.
		
		In order to successfully explore input graphs where the uncolored neighbors are all leaves, \alg thus has to prefer uncolored neighbors over colored neighbors.
	\end{proof}
	\begin{lemma}[General Algorithm Cannot Go Back Unless Necessary]\hfill\label{lem:cantgoback}
		Let \alg be an algorithm that successfully explores all graphs. If the agent of \alg is on a vertex $v$ with only one colored neighbor $p$, which is its predecessor, and at least one uncolored neighbor $u$, and if $p$ has an uncolored neighbor $v' \neq v$, then the algorithm has to choose to go to an uncolored neighbor of $v$.
	\end{lemma}
	
	\begin{proof}
		First, note that this situation is more delicate than the situation of \cref{obs:cantignoreuncolored} since we do not assume that $v$ is colored. Assume towards contradiction that there is an algorithm where the lemma is not true. In other words, assume that there is an algorithm \alg that explores all graphs and where there is a graph with a vertex $v$, a colored predecessor $p$ as its only colored neighbor and at least one uncolored neighbor $u$ and assume that $p$ has at least one uncolored neighbor $v' \neq v$ and \alg does not choose to go to one of the uncolored neighbors while on $v$.
		
		We change the part of the graph on which \alg runs such that $v'$ has as many uncolored neighbors as $v$. Moreover, we assume that all the uncolored neighbors of $v$ and $v'$ are leaves and there are no paths from the rest of the graph to either $v$ or $v'$. This situation is depicted in \cref{fig:cantgoback}.
		
		\begin{figure}[htbp]
			\begin{center}
				\begin{tikzpicture}[x=0.8cm,y=0.8cm,vertex/.style={draw,circle, inner sep=.0pt, minimum size=0.8cm}]
					\node[vertex] (a) at (0,0) {$p$};
					\node[vertex] (b1) at (1.5,0) {$v$};
					\node[vertex] (b2) at (0,-1.5) {$v'$};
					\node[vertex] (c11) at (3,0) {$u$};
					\node[vertex] (c21) at (0,-3) {$u'$};
					
					\draw (a) -- (b1) -- (c11);
					\draw (a) -- (b2) -- (c21);
				\end{tikzpicture}
			\end{center}
			\caption{If the graph that \alg processes contains this structure, \alg cannot explore it successfully.}\label{fig:cantgoback}
		\end{figure}
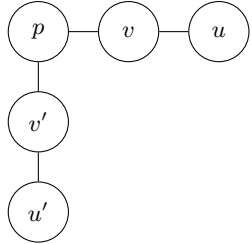
		
		By assumption, \alg goes back to $p$ from $v$.
		If \alg has not colored $v$, then it will again go to $v$ from $p$ and there is an adversary such that \alg never terminates. Hence, \alg colors $v$. At some point, \alg has to visit $v'$, which is only possible from $p$. 
		Whenever the algorithm now wants to visit some uncolored neighbor in $p$, the adversary sends the agent to $v'$.
		When the agent arrives in $v'$, it is in the same situation as before when it was on $v$. Therefore, \alg assigns to $v'$ the same color as to $v$ and then it goes back to $p$. Now, however, due to \cref{obs:functionalnature}, if \alg decides at some point to go to $v$ or to $v'$ from $p$, \alg will, when in $p$, always go to $v$ or to $v'$ and will, therefore, either not be able to explore $u$ and $u'$ or will not return to the start vertex and terminate.
	\end{proof}	
	
	\begin{lemma}\label{lem:cantleaveuncolored}
		Let \alg be an algorithm that successfully explores all graphs. Assume the agent has just moved from a (now) colored vertex $p$ to an uncolored vertex $v$, and that $p$ has at least one other uncolored neighbor $v'$. Then the agent must color $v$ now.
	\end{lemma}
	\begin{proof}
		Assume towards contradiction that there is an algorithm \alg and a graph $G$ with vertices $p$, $v$ and $v'$ such that the described situation arises but $v$ is not colored.
		By changing the graph $G$, we may assume that $p$, $v$, and $v'$ form a triangle and $v'$ is not connected to any other vertices. Moreover, assume that all uncolored neighbors of $v$ apart from $v'$ are leaves and call them $u_1$ to $u_r$, for some $r \in \N$. Denote by $c_1$ to $c_k$, for some $k\in \N$, the colored neighbors of $v$ apart from its predecessor. This situation is depicted in \cref{fig:cantleaveuncolored}.
		
		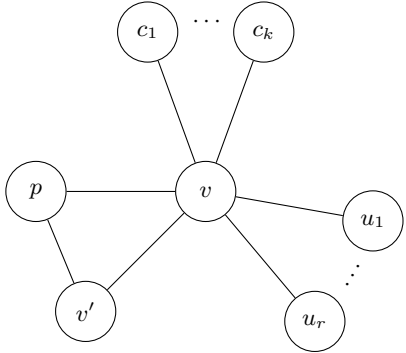
\begin{figure}[htbp]
			\begin{center}
				\begin{tikzpicture}[x=0.8cm,y=0.8cm,vertex/.style={draw,circle, inner sep=.0pt, minimum size=0.8cm}]
					\node[circle, minimum size = 4.5cm] (c) at (0,0){};
					\node[vertex] (p) at (c.west) {$p$};
					\node[vertex] (v) at (c.center) {$v$};
					\node[vertex] (v2) at (c.south west) {$v'$};
					\node[vertex] (u1) at (c.-10) {$u_1$};
					\node[vertex] (ur) at (c.-50) {$u_r$};
					\node[vertex] (c1) at (c.110) {$c_1$};
					\node[vertex] (ck) at (c.70) {$c_k$};
					
					\draw (p) -- (v) -- (v2) -- (p);
					\draw (v) -- (u1);
					\draw (v) -- (ur);
					\draw (v) -- (c1);
					\draw (v) -- (ck);
					
					\node at (c.89) {$\cdots$};
					\node[rotate=60] at (c.-29) {$\cdots$};	
				\end{tikzpicture}
			\end{center}
			\caption{If the vertices $p$, $v$, $v'$ are arranged in a triangle, \alg cannot explore the graph successfully.}\label{fig:cantleaveuncolored}	
		\end{figure}
		
		When \alg leaves $v$, it has three options. Either \alg goes back to $p$, or \alg goes to an uncolored neighbor, or \alg goes to some other colored neighbor $c$. 
		The first option is prohibited by \cref{lem:cantgoback}.
		In case of the second option, the adversary can send \alg to $v'$. We change the graph such that $v'$ has the same neighbors as $v$. Then, the input for \alg will be the same as in the step before, resulting in an endless loop.
		In case of the third option, \alg has to return at some point to $v$ or to $p$ since $v'$ has not been visited yet. 
		Consider the first time when \alg returns to $v$ or to $p$. 
		When \alg wants to visit $v'$ from $p$, the adversary sends \alg to $v$. 
		When \alg is in $v$, the situation looks exactly like before; hence, \alg will again choose the third option and go to the same colored neighbor $c$. 
		Therefore, \alg never visits $v'$ and thus does not successfully explore all graphs, contradicting the assumption. 
	\end{proof}
	
	Before we now state and prove that \textscalt{DepthFirstSearch} 
	uses the minimal number of colors necessary 
	in general, we prove the statement separately for graphs of size $2$ and for graphs of size $3$.
	In fact, we prove a stronger statement, namely without the condition that the algorithm needs 
	to explore all graphs. Hence, we can even assume that the algorithm knows the size of the graph.
	\begin{lemma} \label{lem:nonuniform_lowerbound_size-2-and-3}
		For $k \in \{2,3\}$, the following holds: For every algorithm \alg that successfully explores all graphs of size $k$, there is a graph $G$ of size $k$ such that \alg uses at least $k-1$ colors to explore $G$.
	\end{lemma}
	\begin{proof}
		Clearly, an algorithm needs at least $1$ color to explore a path of two vertices.
		
		For graphs on three vertices and $2$ colors, 
		consider the exploration of a triangle with vertices $x, y, z$. 
		Without loss of generality, assume the agent starts on $x$.
		If the agent does not color $x$ and visits the next vertex, it is caught in an endless loop
		and will always leave the current vertex uncolored and go to the next vertex.
		Therefore, $x$ receives some color $c$ and the agent continues to the next vertex,
		which we call $y$.
		We show now first that the agent assigns some color to $y$ and
		second that the assigned color cannot be $c$.
		
		First, if the agent does not color $y$ and goes back to $x$, the agent then has to visit $z$.
		However, since the agent cannot distinguish $y$ from $z$, the adversary can arrange the ordering
		of the vertices in such a way that the agent visits $y$ instead of $z$, whereupon the agent is
		stuck in an endless loop. Hence, the agent has to color $y$.
		
		Second, if the agent colors $y$ with the same color as $x$ 
		and then 
		continues to $z$---either directly or via $x$---, the agent then cannot distinguish $x$ from $y$. 
		Hence, the adversary can arrange the ordering of the vertices such that the agent will not
		terminate on the start vertex. 
	\end{proof}

	\begin{theorem}\label{thm:generallowerbound}
		For every algorithm \alg that successfully explores all graphs and for every natural number $n$ there is a graph $G$ of size $n$ such that \alg uses at least $n-1$ colors to explore $G$.
	\end{theorem}
	\begin{proof}
		For $n=1$, there is nothing to prove. For $n \in \{2,3\}$, the statement follows from \cref{lem:nonuniform_lowerbound_size-2-and-3}.
		
		For graphs of size $4$, we can use \cref{thm:trees_lower_bound} and the fact that the algorithm has to explore all graphs---and, in particular, all trees---to prove that an algorithm needs $3$ colors to explore all graphs of size $4$.
		We describe how this works: Consider a cycle of $4$ vertices, say $x, y, z, w$.
		Assume the agent starts on $x$.
		Clearly, if the agent does not color $x$ in the first step, no vertex can ever be colored.
		Assume that $x$ receives color $1$ and the agent moves to $y$.
		By \cref{lem:cantgoback}, the agent moves to $z$ next.
		If $y$ were not colored, the next vertex---which is $z$---would again receive color $1$,
		and the agent would not be able to distinguish $x$ from $z$ when standing on $w$.
		Therefore, $y$ receives color $2$ and the agent moves to $z$.
		
		By \cref{lem:cantgoback}, the agent moves to $w$ next.
		If $z$ were not colored and if the graph were a cycle on $5$ vertices with a vertex $u$ between $w$ 
		and $x$, the next vertex, $w$, would again receive color $1$ and, as before, the agent would not be 
		able to distinguish $w$ from $x$ when standing on $u$.
		Since the algorithm is a general algorithm that has to successfully explore all graphs and that does not
		know the size of the graph, vertex $z$ thus has to be colored.
		
		If $z$ received color $2$, and if the graph were a path of seven vertices---as in the proof of \cref{thm:trees_lower_bound}---in the order $w, x, y, z, u, v$, then the vertex $u$ would receive color $2$ as well and on the way back to the start vertex $x$, the agent could not distinguish between $y$ and $u$ when standing on $z$.
		Therefore, $z$ has to receive a third color.
		
		We now turn to the proof for graphs of size $5$ or larger:	
		Assume towards contradiction that there is an algorithm \alg that, for some $n_0 \ge 5$, uses at most $n_0-2$ colors on every graph of size $n_0$.
		We are going to define a family of graphs of size $2(n_0-1)-1$ and prove that there is a graph $G_0$ in this family such that \alg uses at least $n_0-1$ colors on its first $n_0-1$ steps in order to successfully explore $G_0$.
		Afterwards, we define a graph $G_1$ of size $n_0$ which locally looks, for the first $n_0-1$ steps, exactly as $G_0$. Therefore, \alg uses at least $n_0-1$ colors on $G_0$ as well.
		
		We construct our family of graphs as follows.
		Consider a graph $G$ that consists of a path $v_1,\ldots,v_{n_0-1}$ of length $n_0-1$ and where, for each $r\in [n_0-1]$, a leaf $l_r$ is connected to $v_r$.
		For fixed $i,j \in [n_0-1]$, we merge the leaves $l_i$ and $l_j$ and call the new vertex $l_{i/j}$ and the graph $G_{n_0,i,j}$.
		Such a graph is depicted in \cref{fig:speciallowerbound}.

		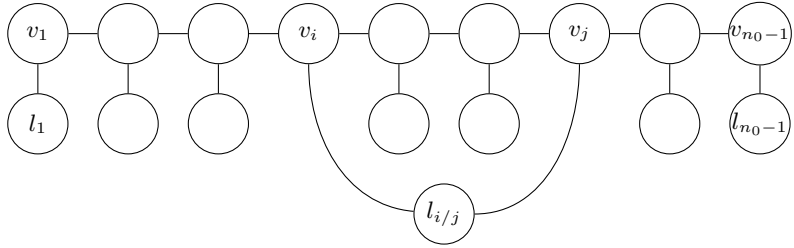
\begin{figure}[htbp]
			\begin{center}
				\begin{tikzpicture}
					\tikzstyle{vertex}=[x=1.2cm, y=1.2cm, draw,circle, inner sep=.0pt, minimum size=0.8cm]
					
					\foreach \name/\x in {2/2, 3/3, 5/5,
						6/6, 8/8}
					\node[vertex] (G-\name) at (\x,0) {};
					
					\foreach \name/\x in {12/2, 13/3, 15/5,
						16/6, 18/8}
					\node[vertex] (G-\name) at (\x,-1) {};
					
					\node[vertex] (G-1) at (1,0) {$v_1$};
					\node[vertex] (G-11) at (1,-1) {$l_1$};
					\node[vertex] (G-4) at (4,0) {$v_i$};
					\node[vertex] (G-7) at (7,0) {$v_j$};
					\node[vertex] (G-9) at (9,0) {$v_{n_0-1}$};
					\node[vertex] (G-19) at (9,-1) {$l_{n_0-1}$};
					\node[vertex] (G-20) at (5.5,-2) {$l_{i/j}$};
					
					\foreach \from/\to in {1/2,2/3,3/4,4/5, 5/6, 6/7, 7/8, 8/9, 1/11, 2/12, 3/13, 5/15, 6/16, 8/18, 9/19}
					\draw (G-\from) -- (G-\to);
					
					\draw (G-4) edge[in=175,out=-90] (G-20);
					\draw (G-7) edge[in=0,out=-90] (G-20);
				\end{tikzpicture}
			\end{center}
			\caption{The graph $G_{n_0,i,j}$} \label{fig:speciallowerbound}
		\end{figure}
		
		Consider any general algorithm \alg with at most $n_0-2$ colors that explores a graph $G_{n_0,r,s}$ for yet to be determined $r$ and $s$.
		The adversary makes the decisions in such a way that \alg walks from $v_1$ to $v_{n_0-1}$ without exploring any leaves and without leaving any visited vertex uncolored. This is possible by \cref{lem:cantgoback} and \cref{lem:cantleaveuncolored}.
		
		By the pigeonhole principle, there are $i<j$ such that \alg assigns the same color to $v_i$ and $v_j$. Let $r=i$ and $s=j$; in other words, consider the exploration of \alg on $G_{n_0,i,j}$.
		Without loss of generality, we assume that $v_i$ and $v_j$ are not adjacent since, if that were the case, clearly, all $v_k$ with $k>j$ would be assigned the color of $v_j$ and \alg cannot explore all graphs this way.
		
		On its way back to $v_1$, \alg reaches $v_j$ at some point. From there, \alg has to visit $l_{i/j}$ by \cref{obs:cantignoreuncolored}.
		Since \alg cannot distinguish $v_i$ and $v_j$ when located in $l_{i/j}$, the adversary can send \alg to either vertex if \alg wants to visit a neighbor.
		If \alg decides not to color $l_{i/j}$, then the adversary sends \alg to $v_i$ and \alg will have to go to $l_{i/j}$ by \cref{obs:cantignoreuncolored} and then either color $l_{i/j}$ or be caught in an endless loop.
		Hence, \alg colors $l_{i/j}$ and then goes to one of its neighbors. The adversary is going to send \alg to $v_i$.
		
		When \alg arrives in $v_i$, the closed neighborhood $N[v_i]$ is colored.
		If \alg decides to go to $l_{i/j}$, it will then again be sent to $v_i$ and will never terminate.
		If \alg decides to go to $v_{i-1}$, it will never visit any leaf between $v_i$ and $v_j$.
		If \alg decides to go to $v_{i+1}$, then \alg can never terminate since all paths from $v_{i+1}$ to $v_1$ lead through $v_i$ and from there, \alg always chooses $v_{i+1}$.
		Therefore, no matter how \alg behaves, it cannot successfully explore $G_{n_0,i,j}$, unless it uses at least $n_0-1$ colors on the first $n_0-1$ steps.
		
		\begin{figure}[htbp]\begin{center}
				\begin{tikzpicture}[x=0.8cm,y=0.8cm,vertex/.style={draw,circle,
				 inner sep=.0pt, minimum size=0.8cm}] 
					\node[vertex] (u) at (4,-4) {$u$};
					\node[vertex] (v_1) at (0,-2) {$v_{1}$};
					\node[vertex] (v_2) at (1.5,-2) {$v_{2}$};
					\node (v_3) at (3,-2) { };
					\node (v_4) at (4.5,-2) { };
					\node[vertex] (v_n_1) at (6,-2) {$v_{n_0-2}$};
					\node[vertex] (v_n) at (7.5,-2) {$v_{n_0-1}$};
					
					\node at (4,-2) {$\cdots$};
					\draw (v_1) -- (u) -- (v_2);
					\draw (u) -- (v_n_1);
					\draw (v_n) -- (u);	
					\draw (v_1) -- (v_2) -- (v_3);
					\draw (v_4) -- (v_n_1) -- (v_n);
				\end{tikzpicture}
				
				\caption{The graph $G_1$ for the general lower bound}\label{fig:generallowerbound}
		\end{center}\end{figure}
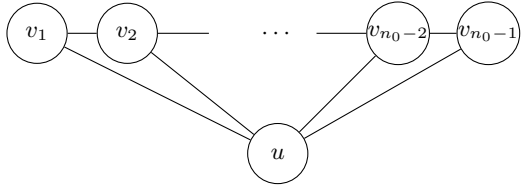
		Consider now the graph $G_1$ depicted in \cref{fig:generallowerbound}. $G_1$ consists of a path of $n_0-1$ vertices that are connected to a universal vertex $u$. Locally, $G_1$ looks exactly the same as $G_{n_0,i,j}$. An algorithm can only successfully explore $G_{n_0,i,j}$ if it assigns $n_0-1$ colors on the first $n_0-1$ steps during the exploration of $G_{n_0,i,j}$.
		Then, however, such an algorithms assigns $n_0-1$ colors on the first $n_0-1$ steps during the exploration of $G_1$ as well. 
		$G_1$ has $n_0$ vertices; hence, we have discovered a graph on $n_0$ vertices where \alg uses at least $n_0-1$ colors. This contradicts our assumption and thus finishes the proof.
	\end{proof}
	
	\subsection{Non-Uniform Algorithms}
	Clearly, the proof of \cref{thm:generallowerbound} relies heavily on the uniformity of the algorithm, that is, the property that the algorithm has to successfully explore every graph.
	If an algorithm just has to explore all graphs of a fixed size $n$ or up to a fixed size $n$, then the situation is different.
	We cannot argue anymore that an algorithm needs to use many colors on a certain small graph by claiming that the algorithm has to use many colors to successfully explore a certain large graph and showing that it is not possible to distinguish the large graph from the small graph.
	Instead, we have to assume that the algorithm knows $n$ because we want to prove a property for all possible algorithms.
	
	We want to explore this restriction in order to know how useful it is to know the size of the graph in our model. 
	We provide almost matching upper and lower bounds for this case in the next two theorems.
	\begin{theorem}
		For every $n\ge 5$ there is an algorithm that successfully explores all graphs of size exactly $n$ and that uses on each such graph at most $n-2$ colors.
	\end{theorem}
	
	\begin{proof}
		Let $n \in \N$, $n\ge5$.
		Our algorithm $\textscalt{DFS}_n$ works as follows.
		In general, $\textscalt{DFS}_n$ works exactly as 
		\textscalt{DepthFirstSearch}.
		However, we want to assign to the last two vertices color $n-2$; therefore, we need a few adaptations that we will describe in a moment.
		
		Denote the last three vertices that 
		\textscalt{DepthFirstSearch} colors by $x$, $y$, and $z$ 
		(in this order).
		As long as the agent is not on an uncolored vertex where 
		there is a neighbor with color $n-2$, $\textscalt{DFS}_n$ 
		works exactly as \textscalt{DepthFirstSearch}.
		Therefore, as long as $\textscalt{DFS}_n$ is not on one of 
		the vertices $y$ or $z$, the behavior of 
		$\textscalt{DFS}_n$ is exactly like the behavior of 
		\textscalt{DepthFirstSearch} and there is always a unique 
		predecessor that can be recognized by the agent.
		
		The last three vertices can be arranged in three different 
		ways, as illustrated in \cref{fig:nonuniformcases}. We 
		adapt \textscalt{DepthFirstSearch} as follows:
		\begin{description}
			\item[Rule 1] If the current vertex is uncolored and 
			there is exactly one neighbor with color $n-2$, then 
			$\textscalt{DFS}_n$ assigns \emph{the smallest color 
			for which there is no neighbor with color number 
			exactly $1$ larger}\footnote{For example, unless $y$ is 
			connected to the vertex $2$, $y$ is assigned color $1$; 
			otherwise, unless $y$ is connected to both the vertex 
			$2$ and the vertex $3$, $y$ is assigned color $2$ etc. 
			Note that this excludes $y$ receiving color $n-3$ since 
			it is connected to $x$. Furthermore, note that $y$ 
			might be assigned color $n-2$.}---which we will call 
			\emph{smallest unproblematic color}--- and continues to 
			an uncolored neighbor. If there is no such uncolored 
			neighbor, then the agent returns to the neighbor with 
			color $n-2$.
			\item[Rule 2] If the current vertex is uncolored and 
			there are two neighbors with color $n-2$ (in this case, 
			$\textscalt{DFS}_n$ is on $z$ and just has to go back 
			to the start vertex), then $\textscalt{DFS}_n$ assigns 
			the smallest unproblematic color and then continues to 
			a vertex of minimal color.
			\item[Rule 3] If the current vertex $v$ is colored with 
			$c(v)>1$ and there is neither an uncolored neighbor nor 
			a neighbor $v'$ with $c(v')=c(v)-1$, then 
			$\textscalt{DFS}_n$ goes to the neighbor with the 
			largest color.
		\end{description}
		
		\begin{figure}[htbp]\begin{center}
				\begin{subfigure}{.3\textwidth}
					\begin{tikzpicture}
						\tikzstyle{vertex}=[x=1.2cm, y=1.2cm, draw,circle, minimum size=0.8cm]
						
						\node[vertex] (x) at (0,0) {$x$};
						\node[vertex] (y) at (1,0.8) {$y$};
						\node[vertex] (z) at (1,-0.8) {$z$};
						
						\draw (x) -- (y) ;
						\draw (x) -- (z);
					\end{tikzpicture}
					\caption{Case 1}
				\end{subfigure}
				\begin{subfigure}{.3\textwidth}
					\begin{tikzpicture}
						\tikzstyle{vertex}=[x=1.2cm, y=1.2cm, draw,circle, inner sep=.0pt, minimum size=0.8cm]
						
						\node[vertex] (x) at (0,0) {$x$};
						\node[vertex] (y) at (1,0.8) {$y$};
						\node[vertex] (z) at (1,-0.8) {$z$};
						
						\draw (x) -- (y) -- (z);
						\draw (x) -- (z);
					\end{tikzpicture}
					\caption{Case 2}
				\end{subfigure}
				\begin{subfigure}{.3\textwidth}
					\begin{tikzpicture}
						\tikzstyle{vertex}=[x=1.2cm, y=1.2cm, draw,circle, inner sep=.0pt, minimum size=0.8cm]
						
						\node[vertex] (x) at (0,0) {$x$};
						\node[vertex] (y) at (1,0) {$y$};
						\node[vertex] (z) at (2,0) {$z$};
						
						\draw (x) -- (y) ;
						\draw (y) -- (z);
					\end{tikzpicture}
					\caption{Case 3}
				\end{subfigure}
				\caption{The three cases how $x$, $y$, and $z$ can be connected}\label{fig:nonuniformcases}
		\end{center}\end{figure}
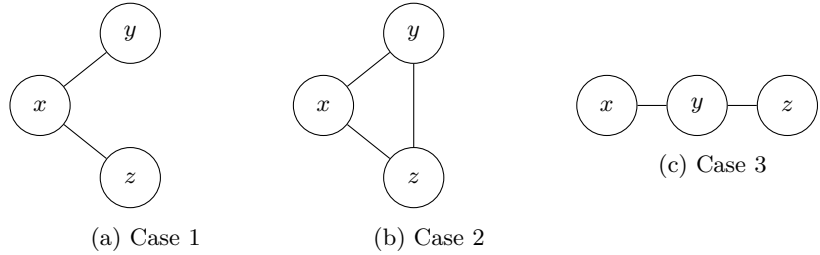
		
		For the behavior of $\textscalt{DFS}_n$ on $y$ and $z$, we 
		have the following cases depending on which color $x$ 
		receives and how the vertices $x$, $y$, and $z$ are 
		connected to each other. We have three interesting cases 
		and two boring cases and we start our analysis with the two 
		boring cases.
		
		First, consider the case where $x$, which is the third-last 
		vertex to be colored, receives color $n-4$ or less. Thus, 
		\textscalt{DepthFirstSearch} never assigns a color larger 
		than $n-2$ and never arrives at a vertex already colored 
		with $n-2$. In this case, $\textscalt{DFS}_n$ behaves 
		exactly as \textscalt{DepthFirstSearch} and uses at most 
		$n-2$ colors as well.
		
		In the second case, $x$ receives color $n-3$. Thus, 
		\textscalt{DepthFirstSearch} may assign to $y$ color $n-2$ 
		and to $z$ color $n-1$. $\textscalt{DFS}_n$ assigns to $y$ 
		color $n-2$ as well, but when it arrives on $z$ and sees 
		exactly one vertex of color $n-2$, $\textscalt{DFS}_n$ uses 
		Rule 1 and assigns the smallest unproblematic color and 
		then goes to the vertex with color $n-2$, which is $y$. 
		On $y$, $\textscalt{DFS}_n$ sees that $y$ was colored in an 
		earlier step and that there are no more uncolored neighbors 
		and thus applies the normal backtracking rule to go back to 
		the vertex with color $n-3$, which is $x$. From there, the 
		behavior of $\textscalt{DFS}_n$ is the same as the behavior 
		of \textscalt{DepthFirstSearch}. 
		
		Note that the color of $z$ cannot interfere with the 
		behavior of \textscalt{DepthFirstSearch} since 
		\textscalt{DepthFirstSearch} always goes back to the direct 
		predecessor, which is a vertex with color number exactly 
		$1$ less.
		
		We can now move towards the interesting cases. In all of the remaining three cases, $x$ receives color $n-2$. The cases are illustrated in \cref{fig:nonuniformcases}.
		
		We describe the behavior of $\textscalt{DFS}_n$ on each of 
		these cases and thus show that $\textscalt{DFS}_n$ can 
		distinguish these cases and can consequently successfully 
		explore graphs of size $n$ with $n-2$ colors.
		
		\begin{description}
			\item[Case 1:] $x$ receives color $n-2$; $x$ is adjacent to both $y$ and $z$; and $y$ and $z$ are not adjacent, see \cref{fig:nonuniformcases}.
			In this case, $y$ is going to be assigned by 
			$\textscalt{DFS}_n$ the smallest unproblematic color 
			and then the agent returns to $x$ (Rule 1). From there, 
			the agent visits $z$, which is the remaining uncolored 
			vertex. On $z$, $\textscalt{DFS}_n$ is in the same 
			situation as it was on $y$ (exactly one neighbor with 
			color $n-2$) and is going to assign the smallest 
			unproblematic color to $z$ and then return to $x$ by 
			Rule 1. Now, $\textscalt{DFS}_n$ behaves again exactly 
			the same as \textscalt{DepthFirstSearch}.
			\item[Case 2:] $x$ receives color $n-2$ and $x$, $y$, $z$ form a triangle.
			Again, $y$ receives the smallest unproblematic color by Rule 1. Then, by Rule 1 as well, the agent does not return to $x$ but continues to the next unvisited vertex, which is $z$.
			On $z$, we have two cases. 
			\begin{description}
				\item[Subcase 2.1:] If $y$ received a color 
				different from $n-2$, then the agent is in the same 
				situation as before (exactly one neighbor with 
				color $n-2$) and assigns the smallest unproblematic 
				color and returns to $x$. From there, 
				$\textscalt{DFS}_n$ behaves again exactly the same 
				as \textscalt{DepthFirstSearch} and terminates as 
				expected.
				\item[Subcase 2.2:] If $y$ received color $n-2$, the situation is more delicate. The agent now sees two vertices colored $n-2$ and if it would assign color $n-2$ and go to $y$, it might be caught in an endless loop and go back and forth between $y$ and $z$. Hence, by Rule 2, the agent assigns the smallest unproblematic color and then goes to the smallest color it sees.
				If the smallest unproblematic color is $k\ge2$, 
				then $z$ is connected with vertices $2,\ldots k$ 
				and the agent goes directly to the start vertex if 
				$z$ is connected to the start vertex or the agent 
				goes to vertex $2$ and from there to the start 
				vertex. If the smallest unproblematic color is $1$, 
				this means that $z$ is not a neighbor of $2$, in 
				which case the agent goes to any vertex with 
				minimal color. If this vertex is $y$, then the 
				agent applies Rule 3 to go to $x$ and from there 
				with normal backtracking to the start vertex. If 
				this vertex is not $y$, then the agent uses normal 
				backtracking to the start vertex right away. Since 
				$z$ did not receive a color that interferes with 
				backtracking, $\textscalt{DFS}_n$ behaves exactly 
				as \textscalt{DepthFirstSearch}.
			\end{description}
			\item[Case 3:] $x$ receives color $n-2$ and $x$, $y$, $z$ form a simple path.
			Again, $y$ receives the smallest unproblematic color and the agent then continues to $z$. We have to distinguish two cases.
			\begin{description}
				\item[Subcase 3.1:] If $y$ was assigned color 
				$n-2$, then $z$ receives the smallest unproblematic 
				color as in Cases 1 and 2 by Rule 1 and the agent 
				then goes back to $y$. From there, the agent 
				behaves again exactly as 
				\textscalt{DepthFirstSearch} and goes to vertex 
				$n-3$ by the normal backtracking rule. Note that 
				$y$ is in this case connected to vertex $n-3$ since 
				$y$ received $n-2$ as smallest unproblematic color.
				\item[Subcase 3.2:] If $y$ is assigned something 
				less than $n-2$, then $z$ is assigned the largest 
				number the agent sees as is the case with 
				\textscalt{DepthFirstSearch} and then goes to the 
				vertex with this largest number (see line 
				\ref{line:savecolors} of Algorithm~\ref{alg:dfs}). 
				If this vertex is not $y$, we are done since then 
				there is always a clear predecessor until the start 
				vertex is reached. If this vertex is $y$, however, 
				Rule 3 will be applied and the agent goes from $y$ 
				to $x$ and then backtracks from there to the start 
				vertex.
			\end{description}
		\end{description}
		We see that no matter how the last three vertices are 
		arranged, $\textscalt{DFS}_n$ successfully explores the 
		graph. 
		This finishes the proof.
	\end{proof}
		
		\begin{theorem}
			For every algorithm \alg and for every natural number $n$ such that \alg explores all graphs \emph{of size $n$} there is a graph $G$ of size $n$ such that \alg uses at least $n-3$ colors on $G$.
		\end{theorem}
		\begin{proof}
			The theorem is clearly true for $n\le 5$ by \cref{lem:nonuniform_lowerbound_size-2-and-3}: 
			Since we showed that at least two colors are needed to explore all graphs of size $3$, 
			two colors are needed to explore all graphs of size $4$ and $5$ as well.
			To see this, just assume that there is an additional leaf (for graphs of size $4$) 
			or an additional path with two vertices (for graphs of size $5$) attached to the start vertex of 
			the triangle in the proof of \cref{lem:nonuniform_lowerbound_size-2-and-3}.
			
			Assume now towards contradiction that there is a natural number $n\ge 6$ and an algorithm \alg that explores every graph of size $n$ with at most $n-4$ colors.
			Consider the following $\binom{n-3}{2}$ graphs of size $n$:
			All cliques of size $n-3$ where two vertices $v_i$ and $v_j$ both have a leaf and the remaining vertices are all connected to some special vertex $s$.
			Note that every vertex in the clique has degree $n-3$ and the vertices $v_i$ and $v_j$ thus cannot be distinguished from the other vertices in the clique. The special vertex $s$ has degree $n-5$.
			
			Note that we can use \cref{lem:cantgoback} and \cref{lem:cantleaveuncolored} as long as there are enough unvisited vertices available such that their proofs work. In case of \cref{lem:cantgoback}, we need three unvisited vertices;
			in case of \cref{lem:cantleaveuncolored}, we need two unvisited vertices including the current vertex.
			Therefore, by \cref{lem:cantgoback} and \cref{lem:cantleaveuncolored},
			an adversary can ensure that \alg first visits all vertices in the clique and then goes to $s$.
			
			Now, \alg will have assigned to two vertices the same color.
			Consider the exploration of \alg on the graph where these two vertices are the two vertices with the attached leaf, $v_i$ and $v_j$.
			Note that here, we need the fact that all vertices in the clique have the same degree. Otherwise, \alg might recognize $v_i$ and $v_j$ and ensure that these two always receive a different color.
			After visiting all vertices in the clique, \alg has to visit the first leaf at some point.
			Say the agent first visits the leaf of $v_i$.
			If \alg does not color the leaf, then it is caught in an endless loop. Hence, \alg colors the leaf and then goes back to $v_i$. 
			From there, \alg has to visit $v_j$ at some point.
			
			Before visiting $v_j$, \alg may visit $s$, but whenever the agent tries to reach $v_j$, the adversary can send the agent to $v_i$.
			The only way to reach $v_j$ is by going directly from $v_i$ to $v_j$.
			Once the agent does this, it can then visit the leaf of $v_j$, color it with the same color as the leaf of $v_i$ due to its functional nature, and then return to $v_j$.
			Now, the agent is in the same situation as before when it went from $v_i$ to $v_j$.
			Therefore, the agent will now go to $v_i$ and from there again to $v_j$ and is hence caught in an endless loop.
			This contradicts our assumption and thus finishes the proof.
		\end{proof}

		\section{Between Trees and General Graphs}\label{sec:circumference}
		
		We want to discover the parameter that determines the difficulty of graph exploration of a given graph.
		Since the exploration of trees is quite easy, it is tempting to think that a good parameter to measure the difficulty for graph exploration could be treewidth, feedback vertex set, number of cycles, etc.
		However, considering the difficult instance in \cref{fig:speciallowerbound}, we see that treewidth and feedback vertex set are not suitable parameters. 
		In fact, even restricting the number of cycles to $1$ and the maximum degree to $3$ does not remove the linear amount of colors.
		Bipartite graphs as generalizations of trees might still serve as plausible candidates.
		After all, the proof of \cref{lem:cantleaveuncolored} does not work for bipartite graphs and the graphs $G_{n_0,i,j}$ used in the construction for \cref{thm:generallowerbound} are not bipartite if $j-i$ is odd.
		However, the following theorem shows how these defects can be remedied; in fact, only exploring bipartite graphs is almost as difficult as exploring all graphs.
		\begin{theorem}\label{thm:bipartitegraphslowerbound}
			There is no algorithm \alg and no natural number $n$ such that \alg successfully explores all bipartite graphs and uses strictly fewer than $n-2$ colors on every graph on $n$ vertices.
		\end{theorem}
		\begin{proof}
			Assume towards contradiction that such an algorithm \alg and such a number $n_0$ exists.
			As in the proof of \cref{thm:generallowerbound}, we let \alg walk on graphs of the family $G_{N,i,j}$ with $j-i$ even and for $N=3n_0$.
			By \cref{lem:cantgoback}, we can ensure that \alg walks on the path without exploring any leaves.
			Since \alg uses at most $n_0-2$ colors on every graph of size $n_0$, \alg uses a color twice on its first $n_0-1$ steps, say $c(v_l)=c(v_m)$ for $l<m$. Then, however, \alg is going to repeat its entire pattern after $v_l$, so for all $k\in \naturalnumberpositive$:
			\[c(v_{l+k})=c(v_{m+k}).\]
			Since our graph is large enough, we can ensure that this pattern is repeated at least three times.
			This means that we can find two vertices of the same color that are in the same partition.
			We want to apply the reasoning of \cref{thm:generallowerbound} and hence connect the vertices as follows to ensure that the graph stays bipartite:
			If the distance $(m-l)$ of $v_l$ and $v_m$ is even, we define $v_{i} \coloneqq v_{l+1}$ and $v_{j} \coloneqq v_{m+1}$ and connect these vertices via a leaf $l_{i/j}$.
			If the distance $(m-l)$ is odd, we define $v_{i} 
			\coloneqq v_{l+1}$ and $v_{j} \coloneqq v_{m+1+(m-l)}$ 
			and connect them via a leaf $l_{i/j}$.
		
		Now, in order to apply the same reasoning as in the proof of \cref{thm:generallowerbound}, we have to guarantee that the entire neighborhood of $v_i$ and $v_j$ is colored once \alg visits and colors $l_{i/j}$.
		To ensure this, we prove that no vertex in the repeating pattern $v_l,\ldots, v_{m-1}$ is left uncolored: Assume that any vertex in the repeating pattern would be left uncolored, say $c(v_{m-1})=0$.
		Clearly, \alg cannot leave two consecutive vertices uncolored; otherwise, \alg would leave all future vertices on the path uncolored.
		Therefore, $v_{m-2}$ is colored and we denote its color by $x$.
		On vertex $v_{m-1}$, \alg received as input $(0,x,0,0)$ since there are two uncolored neighbors and one neighbor with color $x$.
		If $m-l$ is even, connect vertex $v_{m+(m-l)}$ with $v_{m-2}$ instead of $v_{m+1+(m-l)}$. If $m-l$ is odd, connect $v_{m+2(m-l)}$ with $v_{m-2}$ instead of $v_{m+1+2(m-l)}$. Assume without loss of generality that $m-l$ is even and observe what happens when the algorithm reaches $v_{m-1+(m-l)}$: This vertex is left uncolored as $v_{m-1}$ and the algorithm then visits $v_{m+(m-l)}$.
		Upon arrival in $v_{m+(m-l)}$, the algorithm receives $(0,x,0,0)$ as input and therefore decides not to color $v_{m+(m-l)}$ and instead goes to an uncolored neighbor. The adversary arranges the port labels such that $v_{m-1+(m-l)}$ is chosen. Now, \alg is in an endless loop. Therefore, all vertices of the repeating pattern receive a color.
		
		Hence, the entire neighborhood of $v_i$ and $v_j$ is colored once \alg visits $l_{i/j}$ and we can now apply the same reasoning as in the proof of \cref{thm:generallowerbound} to conclude that \alg may never repeat its pattern.
		However, this time, the smallest bipartite graph that looks locally like some $G_{n_0,i,j}$ for the first $n_0-1$ steps is not exactly $G_1$ from \cref{fig:generallowerbound}, but $G_{\text{bipartite}}$, where $u$ is only connected to all odd vertices $v_h$ and an additional vertex $u_{\text{even}}$ is connected to all even path vertices, leading to a lower bound of $n-2$.
	\end{proof}
	
	We now argue why the \emph{circumference} of a graph, i.e., the size of a longest simple cycle, is the right parameter.
	First, we observe that the construction from the proof of \cref{thm:generallowerbound} immediately results in the following corollary:
	\begin{corollary} \label{cor:parameterizedlowerbound}
		Any algorithm needs at least $k-2$ colors to color all graphs of circumference at most $k$.
	\end{corollary}
	Second, we present with Algorithm~\ref{alg:smalldfs} an algorithm that uses at most $2k-1$ colors and successfully explores all graphs of circumference at most $k\ge 3$.
	This algorithm sorts numbers $a,b$ with respect to whether the 
	distance between $a,\ldots,b$ or the distance between 
	$b,b+1,\ldots,2k-1,1,2,\ldots,a$ is larger.
	We use the following maximum function for three natural numbers $a,b,k$ with $a\le b\le 2k-1$:
	\[{\max}^k\{a,b\} \coloneqq \begin{cases}
		a \text{, if } |b-a+1| < |2k-b+a| \\
		b \text{, if } |b-a+1| > |2k-b+a|. \\
	\end{cases}\]
	Note that, if the two values $(b-a+1)$ and $(2k-b+a)$ were equal, this would imply $2(b-a) = 2k-1$, which is impossible for three natural numbers $a,b$ and $k$.
	In other words, $\maxk$ is well-defined.
	With this maximum function, that is, calculating modulo $2k-1$ shifted by 1, a number is smaller than the next $k-1$ numbers and larger than the preceding $k-1$ numbers. For example, for $k\le 11$, $2k-10$ is smaller than $2k-9,\ldots,2k-1,1,\ldots k-10$ and larger than $2k-11,\ldots, 2k-10-(k-1)=k-9$.
	
	We want to extend $\maxk$ to multiple values $a_1\le\ldots\le a_n \in [2k-1]$, i.e., we would like to define $\maxk\{a_1,\ldots, a_n \}$.
	However, it is not immediately clear how to do that since $\maxk$ is not necessarily transitive. For example, $\maxk\{1,2\}=2$ and $\maxk\{2,k+1\}=k+1$, but still $\max^k\{1,k+1\}=1$. Thus, $\max\{1,2, k+1\}$ could by either $1$, $2$ or $k+1$.
	This problem vanishes if there is an input element (for the maximum function) $a_i$, $i\in [n]$ such that all other input elements $a_j$ are among the $k-1$ elements preceding $a_i$.
	It will become clear later that the construction of algorithm 
	$\textscalt{SmallDFS}_k$ enforces this situation.
	In this case, we use the natural extension of $\maxk$;
	otherwise, we define $\maxk$ to be just the normal maximum function.
	It will become clear later that $\maxk$ is not going to be the ``normal maximum function.''
	
	\begin{algorithm}[htbp]
		\caption{$\textscalt{SmallDFS}_k$}
		\label{alg:smalldfs}
		\textbf{Input:} A graph whose circumference is at most $k$.
		In each step, the input is $c(v) \in \N$, the color of the current vertex $v$, and $c(v_1), \ldots, c(v_d)$, the colors of $v$'s neighbors, where $d$ is the degree of $v$.\\
		\textbf{Output:} In each step, the algorithm outputs $c(v)$, the color the agent assigns to $v$, and $i \in [d]$, the vertex the agent has to move to next.\\
		\begin{algorithmic}[1]
			\If{$c(v)=0$}
			\State $c(v) \coloneqq \left(\left(\underset{v' \in N(v)}\maxk c(v')\right)+1\right) \modi (2k-1)$ \label{line:smalldfscoloring}
			\EndIf
			
			\If{there is an uncolored neighbor $w$} \label{line:smalldfsconditionuncoloredneighbor}
			\State go to $w$
			\ElsIf{there is a neighbor $w$ with $c(w) = \left(c(v)-1\right) \modi (2k-1)$}
			\State go to $w$ \label{line:smalldfsgobackwards}
			\Else
			\State \textbf{terminate}
			\EndIf	
		\end{algorithmic}
	\end{algorithm}
	
	The strategy of $\textscalt{SmallDFS}_k$ is a depth-first 
	search strategy that assigns color $c \modi (2k-1)$ to vertices 
	that are visited in distance $c-1$ (in the depth-first search 
	tree) from the start vertex.
	This is done by looking at all neighbor colors and then assigning the largest color modulo $(2k-1)$ to the current vertex.
	However, in order to be able to assign the color $1$ after the color $2k-1$, the algorithm takes the maximum function defined above to determine the largest color number.
	A problem might arise if the maximum function were not well-defined.
	However, since there are no cycles of length $k+1$, the colors of the current vertex and all its neighbors are always going to be from $k$ sequential numbers in the number range $1, \ldots, 2k-1, 1, \ldots k$. 
	Thus, it is possible, as we are going to prove in a moment, to determine the direct predecessor and backtrack accurately.
	
	\begin{lemma}\label{lem:smalldfsterminates}
		$\textscalt{SmallDFS}_k$ terminates.
	\end{lemma}
	\begin{proof}
		We prove that there are no endless loops, i.e., there is no 
		sequence of vertices $v_1,v_2,\ldots,v_l, v_{l+1}=v_1$ such 
		that $\textscalt{SmallDFS}_k$ always goes from $v_i$ to 
		$v_{i+1 \modi (2k-1)}$. 
		Assume towards contradiction that such a sequence exists. 
		If $\textscalt{SmallDFS}_k$ arrives at a vertex that has 
		been colored before, it goes to an uncolored neighbor 
		regardless of the color of any other vertices.
		Therefore, at some point, there will be no more uncolored neighbors of the vertices in the loop. 
		In other words, the only line that can be applied is line~\ref{line:smalldfsgobackwards}.
		However, to obtain an endless loop, the colors in the loop would then have to contain the sequence $(1,2,\ldots,2k-1)$.
		This is a contradiction to the assumption that the graph has a circumference of at most $k$.
	\end{proof}
	
	\begin{lemma}\label{lem:smalldfsterminatesonstartvertex}
		$\textscalt{SmallDFS}_k$ terminates on the start vertex.
	\end{lemma}
	\begin{proof}
		Since every vertex $v$ except the start vertex has a 
		neighbor $w$ with $c(w) = c(v)-1 \modi (2k-1)$ because of 
		line~\ref{line:smalldfscoloring}, $\textscalt{SmallDFS}_k$ 
		cannot terminate on a vertex that is not the start vertex.
		Since $\textscalt{SmallDFS}_k$ terminates due to 
		\cref{lem:smalldfsterminates}, it has to terminate on the 
		start vertex.
	\end{proof}
	
	We proved that $\textscalt{SmallDFS}_k$ terminates on the start 
	vertex and only on the start vertex. 
	We have yet to show that all vertices are indeed explored.

	\begin{lemma}\label{lem:smalldfspredecessor}
		On any vertex $v$, either $\textscalt{SmallDFS}_k$ goes to 
		an unvisited vertex or it goes back to the direct 
		predecessor of $v$.
	\end{lemma}
	\begin{proof}
		Similar to the proof of \cref{lem:directpredecessors}, we 
		prove something slightly different, namely that (1) 
		$\textscalt{SmallDFS}_k$ always assigns color $c(v')+1 
		\modi (2k-1)$ to a vertex $v$ with predecessor $v'$ and 
		that (2) there is never a vertex $v$ with two neighbors 
		$w_1,w_2$ with $c(w_1)=c(w_2)=c(v)-1 \modi (2k-1)$. This 
		implies that $\textscalt{SmallDFS}_k$ always goes back to 
		the direct predecessor.
		
		Assume towards contradiction that one of these two claims is wrong.
		Consider the first step in which one of these claims is wrong. We have three cases:
		\begin{description}
			\item[Case 1] In this step, $\textscalt{SmallDFS}_k$ 
			does not assign color $c(v')+1 \modi (2k-1)$ to a 
			vertex $v$ with predecessor $v'$.
			This can only happen if there is a neighbor $x$ of $v$ with greater color than $v'$, that is, $\maxk\{c(x),c(v')\}=c(x)$.
			
			Both $x$ and $v'$ are colored before $v$. Assume as 
			subcase 1.1 that $x$ is colored before $v'$. Consider 
			the search sequence $(x,v_1,\ldots,v_n, v')$ of 
			$\textscalt{SmallDFS}_k$ between $x$ and $v'$. The 
			first step when $\textscalt{SmallDFS}_k$ ``skips'' a 
			color is between $v'$ and $v$; hence, the colors 
			between vertices next to each other in this search 
			sequence differ by exactly 1.
			Since $c(x)$ is smaller than the next $k-1$ numbers and since we have $\maxk\{c(x),c(v')\}=c(x)$, there have to be at least $k-1$ vertices between $x$ and $v'$.
			Thus, we have together with $v$ a cycle of length at least $k+2$, which is not allowed.
			Hence, $x$ is not colored first.
			
			Consider subcase 1.2, namely that $v'$ is colored before $x$.
			Then $\textscalt{SmallDFS}_k$ went from $v'$ to $x$, 
			then back to $v'$ and then to $v$. However, this is not 
			possible since backtracking (applying 
			line~\ref{line:smalldfsgobackwards}) is only allowed 
			when there are no uncolored neighbors and thus, 
			$\textscalt{SmallDFS}_k$ cannot backtrack from $x$, 
			which has $v$ as an uncolored neighbor.
			\item[Case 2] In this step, $\textscalt{SmallDFS}_k$ 
			assigns color $c(v)$ to a vertex $v$ with two neighbors 
			$w_1,w_2$ with $c(w_1)=c(w_2)=c(v)-1 \modi (2k-1)$. 
			Assume without loss of generality that first $w_1$ is 
			visited (and colored), then $w_2$, then $v$.
			We can argue exactly as before. The search sequence of 
			$\textscalt{SmallDFS}_k$ between $w_1$ and $w_2$ is 
			$(w_1,v_1,\ldots, v_r, w_2)$. 
			$\textscalt{SmallDFS}_k$ could not backtrack from $w_1$ 
			and thus went to some uncolored vertex $v_1$. 
			Again, backtracking would only be possible up to $w_1$; therefore, we can assume without loss of generality that all vertices $v_1,\ldots,v_r$ in the search sequence are visited for the first time.
			In order for $w_2$ to obtain the same color as $w_1$, there have to be at least $2k-2$ vertices between $w_1$ and $w_2$.
			Since $w_1$ and $w_2$ are both adjacent to $v$, this results in a cycle of size $2k+1$.
			
			\item[Case 3] In this step, $\textscalt{SmallDFS}_k$ 
			assigns color $c(w_1)$ to a vertex $w_1$ adjacent to a 
			vertex $v$ with $c(w_1)=c(v)-1 \modi (2k-1)$ and $v$ is 
			adjacent to a vertex $w_2 \ne w_1$ with $c(w_2)=c(w_1)$.
			Consider the search sequence $(v,\ldots,w_1)$ between $v$ and $w_1$.
			Since $w_1$ receives a greater number than $v$, there have to be at least $k-1$ vertices between $v$ and $w_1$. Together with $v$ and $w_1$, they form a cycle of size $k+1$, which is impossible.
		\end{description}
		
		We see that the assumption that one of these two claims is wrong leads to a contradiction; therefore, both claims are true, which proves the lemma.
	\end{proof}

	\begin{lemma}
		$\textscalt{SmallDFS}_k$ explores all vertices.	
	\end{lemma}
	\begin{proof}
		Since \cref{lem:smalldfspredecessor} is the analogue to 
		\cref{lem:directpredecessors}, this claim can be proved 
		analogously to \cref{lem:dfsexploresallvertices}. 
	\end{proof}
	
	These lemmas together yield the following theorem.
	\begin{theorem} \label{thm:parameterizedupperbound}
		$\textscalt{SmallDFS}_k$ uses at most $2k-1$ colors and is 
		correct on graphs with circumference at most $k$.
	\end{theorem}
	
	Contrasting Corollary~\ref{cor:parameterizedlowerbound} and 
	\cref{thm:parameterizedupperbound}, we see that 
	$\textscalt{SmallDFS}_k$ uses at most $k+1$ too many colors.
	We narrow this gap with the following theorem.
	
	\begin{theorem}
		An algorithm needs at least $2k-3$ colors to successfully 
		color all graphs of circumference at most $k$.
	\end{theorem}
	\begin{proof}
		Let \alg be an algorithm that uses at most $2k-4$ colors and colors every graph of circumference $k$.
		Without loss of generality, assume that \alg uses exactly $2k-4$ colors.
		The colors \alg assigns have to repeat at some point, that is, if we let \alg explore a path $v_1,\ldots,v_r$ of length $r>2k-4$, after some initial color assignments $c(v_1), c(v_2),\ldots, c(v_s)$, the colors \alg assigns always follow the same pattern $c(v_s), c(v_{s+1}),\ldots,c(v_t), c(v_s), c(v_{s+1})$ etc.
		Without loss of generality, we assume this pattern starts with $v_1$.
		This behavior can be enforced even if we add some leaves or connect some vertices of the path:
		Due to \cref{lem:cantgoback}, a general algorithm may not go back to a visited vertex if there is an unvisited neighbor.
		Therefore, we may even add some leaves to the path or connect some vertices of the path and there is still an adversary that ensures that \alg goes from $v_1$ directly to $v_r$ without visiting any leaf (apart from $v_1$ and $v_r$, of course) and without taking any shortcuts.
		
		Consider now the following two graphs.
		Take a path $v_1,\ldots,v_{2k}$ of length $2k-1$ and add a leaf $l_{2k-3}$ to $v_{2k-3}$.
		To obtain $G_1$, add a leaf $l_{k-1}$ to $v_{k-1}$ and connect $v_1$ and $v_k$.
		To obtain $G_2$ instead, add leaves $l_1$ to $v_1$ and $l_k$ to $v_k$ and connect $v_{k-1}$ to $v_{2k-2}$.
		$G_1$ and $G_2$ are depicted in \cref{fig:parameterizedlowerbound}. Note that both $G_1$ and $G_2$ have circumference exactly $k$.
		
		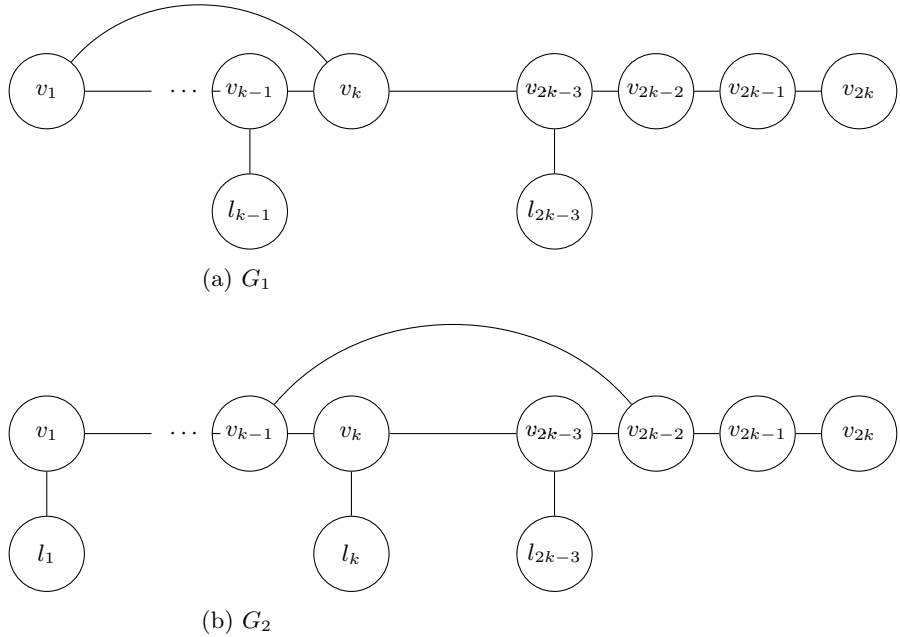
\begin{figure}[htbp]
			\begin{subfigure}{0.5\textwidth}
				\begin{tikzpicture}
					\tikzstyle{vertex}=[x=1.35cm, y=1.6cm, 
					draw,circle, inner sep=2pt, minimum size=1cm]
					
					\node[vertex] (v1) at (1,0) {$v_{1}$};
					\node[minimum size = 0.9cm] (fill1) at (3.2,0) {$\cdots$};
					\node[vertex] (vkminus1) at (3,0) {$v_{k-1}$};
					\node[vertex] (vk) at (4,0) {$v_{k}$};
					\node[minimum size = 0.9cm] (fill2) at (8,0) {$\cdots$};
					\node[vertex] (v2kminus1) at (6,0) {$v_{2k-3}$};
					\node[vertex] (v2k) at (7,0) {$v_{2k-2}$};
					\node[vertex] (v2kplus1) at (8,0) {$v_{2k-1}$};
					\node[vertex] (v2kplus2) at (9,0) {$v_{2k}$};
					
					\node[vertex] (lkminus1) at (3,-1) {$l_{k-1}$};
					\node[vertex] (l2kminus1) at (6,-1) {$l_{2k-3}$};
					
					\draw (v1) -- (fill1) -- (vkminus1) -- (vk) -- (fill2) -- (v2kminus1) --(v2k) --(v2kplus1) -- (v2kplus2);
					\draw (vkminus1) -- (lkminus1);
					\draw (v2kminus1) -- (l2kminus1);
					\draw (v1) edge[bend left = 50] (vk);
				\end{tikzpicture}
				\caption{$G_1$}
			\end{subfigure}
			\newline
			\begin{subfigure}{.5\textwidth}
				\begin{tikzpicture}
					\tikzstyle{vertex}=[x=1.35cm, y=1.6cm, 
					draw,circle, inner sep=2pt, minimum size=1cm]
					
					\node[vertex] (v1) at (1,0) {$v_{1}$};
					\node[vertex] (l1) at (1,-1) {$l_{1}$};
					\node[minimum size = 0.9cm] (fill1) at (3.2,0) {$\cdots$};
					\node[vertex] (vkminus1) at (3,0) {$v_{k-1}$};
					\node[vertex] (vk) at (4,0) {$v_{k}$};
					\node[vertex] (lk) at (4,-1) {$l_{k}$};
					\node[minimum size = 0.9cm] (fill2) at (8,0) {$\cdots$};
					\node[vertex] (v2kminus1) at (6,0) {$v_{2k-3}$};
					\node[vertex] (v2k) at (7,0) {$v_{2k-2}$};
					\node[vertex] (v2kplus1) at (8,0) {$v_{2k-1}$};
					\node[vertex] (v2kplus2) at (9,0) {$v_{2k}$};
					
					\node[vertex] (l2kminus1) at (6,-1) {$l_{2k-3}$};
					
					\draw (v1) -- (fill1) -- (vkminus1) -- (vk) -- (fill2) -- (v2kminus1) --(v2k) --(v2kplus1) -- (v2kplus2);
					\draw (v2kminus1) -- (l2kminus1);
					\draw (vkminus1) edge[bend left = 50] (v2k);
					\draw (v1) -- (l1);
					\draw (vk) -- (lk);
				\end{tikzpicture}
				\caption{$G_2$}
			\end{subfigure}
			\caption{The graphs $G_1$ and $G_2$}\label{fig:parameterizedlowerbound}
		\end{figure}

		Consider the exploration of \alg on $G_1$. As explained, there is an adversary that lets \alg walk from $v_1$ to $v_{2k-2}$ without taking the shortcut $\{v_1,v_k\}$ and without exploring the leaves.
		Moreover, by assuming that the repeating pattern of \alg starts with $v_1$, we obtain in particular that $c(v_1)=c(v_{2k-3})$.
		
		Now, on vertex $v_k$, \alg receives as input $c(v_k)=0$ and the colors $c(v_1)$, $c(v_{k-1})$, and $0$ and assigns color $c(v_k)>0$.
		Then, \alg goes on to explore the rest of the graph. At some point, \alg returns to $v_k$ and receives as input $c(v_k)$ and the colors $c(v_1)$, $c(v_{k-1})$, and $c(v_{k+1})$.
		Since $l_{k-1}$ has not yet been explored and since upon going to $v_1$, \alg has to terminate, it is crucial that \alg chooses $v_{k-1}$ as its next vertex.
		
		However, consider now the exploration of \alg on $G_2$.
		Since $G_2$ looks locally exactly like $G_1$, there is an adversary that lets \alg behave on $G_2$ exactly as on $G_2$ for at least the first $2k-3$ steps.
		On vertex $v_{2k-2}$, \alg receives as input $c(v_{2k-2})=0$ and the colors $c(v_{2k-3})=c(v_1)$, $c(v_{k-1})$, and $0$. This is the same input as before and because of its functional nature, \alg has to assign color $c(v_{2k-2})=c(v_k)>0$ and then has to go to the unvisited neighbor $v_{2k-1}$.
		On $v_{2k-1}$, \alg is in the same situation as during the exploration of $G_1$ on $v_{k+1}$ and assigns color $c(v_{2k-1})=c(v_{k+1})$.
		After then visiting $v_{2k}$, \alg returns to $v_{2k-2}$. Now, its input is again $c(v_k)=c(v_{2k-2})$ and the colors $c(v_1)=c(v_{2k-3})$, $c(v_{k-1})$, and $c(v_{2k-1})=c(v_{k+1})$.
		
		However, this time, since $l_{2k-3}$ has not yet been explored, it is crucial that \alg chooses $v_{2k-3}$ and not $v_{k-1}$. Since \alg is a function, it cannot output different values on the same input; hence, \alg fails to explore either $G_1$ or $G_2$.	
	\end{proof}
	
	\section{Special Graph Classes}\label{sec:P2}
	In the above, we found that trees cannot be explored with less than $3$ colors, general graphs cannot be explored with less than $n-1$ colors and in general, the number of colors needed to explore a graph grows with its circumference. However, if we limit ourselves to special graph classes, we can find algorithms that use much fewer colors than the circumference of the graph.
	
	We illustrate this by taking the example of the graph class $P^2$, the \emph{square paths}.
	The graph class $P^2$ contains all graphs that consist of a path plus edges between all vertices in distance 2 from each other.
	An example is given in \cref{fig:P2}.
	\begin{figure}[htbp]\begin{center}
			\begin{tikzpicture}
				\tikzstyle{vertex}=[x=1.2cm, y=1.2cm, draw,circle, inner sep=.0pt, minimum size=0.7cm]
				
				\foreach \name/\x in {1/1, 3/2, 5/3, 7/4, 9/5, 
					11/6, 13/7, 15/8, 17/9}
				\node[vertex] (G-\name) at (\x,0) {$v_{\name}$};
				
				\foreach \name/\x in {2/1,4/2, 6/3, 8/4, 10/5, 
					12/6, 14/7, 16/8, 18/9}
				\node[vertex] (G-\name) at (\x+0.5,-1) {$v_{\name}$};
				
				\foreach \from/\to in {
					1/2, 2/3,3/4,4/5,5/6,6/7,7/8,8/9,9/10,10/11,11/12,12/13,13/14,14/15,15/16,16/17,17/18,
					1/3, 3/5, 5/7, 7/9, 9/11, 11/13, 13/15, 15/17,
					2/4, 4/6, 6/8, 8/10, 10/12, 12/14, 14/16, 16/18}
				\draw (G-\from) -- (G-\to);
				
				\foreach \from/\to in {
					1/2,2/3,3/4,4/5,5/6,6/7,7/8,8/9,9/10,10/11,
					11/12,12/13,13/14,14/15,15/16,16/17,17/18}
				\draw[very thick] (G-\from) -- (G-\to);
			\end{tikzpicture}
			\caption{An example of a square path $P^2$. The original path is emphasized.}\label{fig:P2}
	\end{center}\end{figure}
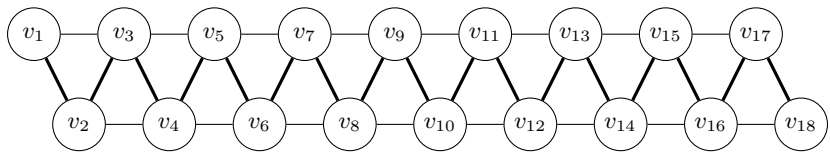
	Note that graphs in $P^2$ have a circumference of size $n$. We show that $4$ colors are enough to explore any graph in $P^2$.
	
	The strategy of our algorithm, 
	\textscalt{SquarePathExploration}, is similar to 
	\textscalt{TreeExploration}, namely to achieve a sense of 
	direction by coloring the vertices alternately $1$, $2$, and 
	$3$.
	For example, if we say that color $1$ is red, color $2$ is green, and color $3$ is blue, and given the graph depicted in \cref{fig:P2_1} with $v_3$ as start vertex, the goal would be to color the vertices as depicted in the figure such that the colors alternate nicely.
	\begin{figure}[htbp]\begin{center}
			\begin{tikzpicture}
				\tikzstyle{vertex}=[x=1.2cm, y=1.2cm, draw,circle, inner sep=.0pt, minimum size=0.7cm]
				
				\foreach \name/\x in {1/1, 3/2, 
					15/8, 17/9}
				\node[vertex] (G-\name) at (\x,0) {$v_{\name}$};
				
				\foreach \name/\x in {2/1, 4/2, 
					14/7, 16/8, 18/9}
				\node[vertex] (G-\name) at (\x+0.5,-1) 
				{$v_{\name}$};
				\foreach \name/\x in {5/3, 11/6}
				\node[vertex, fill=red] (G-\name) at (\x,0) {$v_{\name}$};
				
				\foreach \name/\x in {10/5}
				\node[vertex, fill=red] (G-\name) at (\x+0.5,-1) 
				{$v_{\name}$};
				\foreach \name/\x in {7/4, 13/7}
				\node[vertex, fill=green] (G-\name) at (\x,0) {$v_{\name}$};
				
				\foreach \name/\x in {6/3, 12/6}
				\node[vertex, fill=green] (G-\name) at (\x+0.5,-1) 
				{$v_{\name}$};
				\foreach \name/\x in {9/5}
				\node[vertex, fill=cyan] (G-\name) at (\x,0) {$v_{\name}$};
				
				\foreach \name/\x in {8/4}
				\node[vertex, fill=cyan] (G-\name) at (\x+0.5,-1) {$v_{\name}$};
				
				\foreach \from/\to in {
					1/2, 2/3,3/4,4/5,5/6,6/7,7/8,8/9,9/10,10/11,11/12,12/13,13/14,14/15,15/16,16/17,17/18,
					1/3, 3/5, 5/7, 7/9, 9/11, 11/13, 13/15, 15/17,
					2/4, 4/6, 6/8, 8/10, 10/12, 12/14, 14/16, 16/18}
				\draw (G-\from) -- (G-\to);
				\draw[red, thick, -Stealth] (G-5) edge[out=20, in=160] (G-7);
				\draw[red, thick, -Stealth] (G-7) edge[in=-20,out=-160] (G-5);
				\foreach \from/\to in {5/6, 6/7, 7/8, 8/9, 9/10, 10/11, 11/12}
				\draw[red, thick, -Stealth] (G-\from) -- (G-\to);
			\end{tikzpicture}
			\caption{The desired structure. A possible walk of the agent is depicted with red arrows. The order at the beginning is $v_5, v_7, v_5, v_{6}, v_7$.}\label{fig:P2_1}
	\end{center}\end{figure}
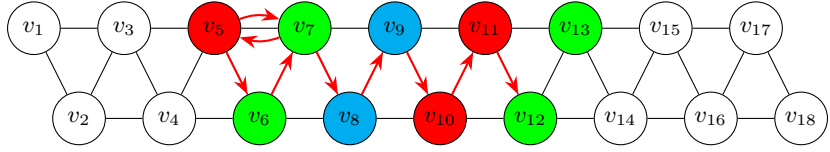
	
	We have to be careful that the agent does not produce a triangle with colors $1$, $2$, and $3$, thus being caught in an endless loop.
	The key idea to achieve this is by exploiting the structure of the square paths:
	Whenever an agent goes to a new vertex $v$ and then sees only one colored neighbor $w$---thus, in general, three uncolored neighbors---, it can go back to~$w$, which then has another uncolored neighbor that is a neighbor of $v$.
	Note that since the algorithm does not have to successfully explore all graphs, but only all square paths, returning back to a visited vertex is easily possible.
	
	However, we have to consider some special cases.
	If we observe the walk of the agent in \cref{fig:P2_1}, we notice two things.
	First, the beginning is quite peculiar.
	The agent starts in $v_5$ and colors it red. Then it visits an uncolored neighbor, $v_7$, colors it green and since there is only one colored neighbor, the agent goes back to this neighbor.
	Back on $v_5$, the agent recognizes that it has been in this vertex before since $v_5$ is colored. 
	Therefore, it does not go back to the only visited---and thus colored---neighbor, but instead goes to an unvisited neighbor, say $v_6$.
	On $v_6$, there are two colored neighbors and usually, the agent would go forward to an unvisited vertex in such a case.
	However, then, it might happen that the vertices on one side of the start vertex are not colored at all.
	Such a situation is depicted in \cref{fig:P2_special1}.
	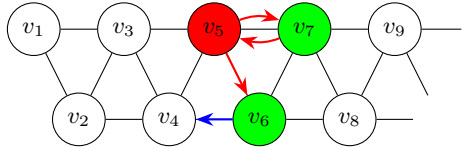
\begin{figure}[htbp]\begin{center}
			\begin{tikzpicture}
				\tikzstyle{vertex}=[x=1.2cm, y=1.2cm, draw,circle, inner sep=.0pt, minimum size=0.7cm]
				
				\foreach \name/\x in {1/1, 3/2, 
					9/5}
				\node[vertex] (G-\name) at (\x,0) {$v_{\name}$};
				
				\foreach \name/\x in {2/1, 4/2, 8/4}
				\node[vertex] (G-\name) at (\x+0.5,-1) 
				{$v_{\name}$};
				\foreach \name/\x in {5/3}
				\node[vertex, fill=red] (G-\name) at (\x,0) {$v_{\name}$};
				
				\foreach \name/\x in {7/4}
				\node[vertex, fill=green] (G-\name) at (\x,0) {$v_{\name}$};
				
				\foreach \name/\x in {6/3}
				\node[vertex, fill=green] (G-\name) at (\x+0.5,-1) {$v_{\name}$};
				
				\foreach \from/\to in {
					1/2, 2/3,3/4,4/5,5/6,6/7,7/8,8/9,
					1/3, 3/5, 5/7, 7/9,
					2/4, 4/6, 6/8}
				\draw (G-\from) -- (G-\to);
				
				\node (G-11) at (7,0) {};
				\node[vertex, white] (G-10) at (5.5,-1) {};
				\draw (G-9) -- (G-11);
				\draw (G-9) -- (G-10);
				\draw (G-8) -- (G-10);
				\draw[red, thick, -Stealth] (G-5) edge[out=20, in=160] (G-7);
				\draw[red, thick, -Stealth] (G-7) edge[in=-20,out=-160] (G-5);
				\draw[red, thick, -Stealth] (G-5) -- (G-6);
				\draw[blue, thick, -Stealth] (G-6) -- (G-4);
			\end{tikzpicture}
			\caption{Special case at the beginning: If the agent went from $v_{6}$ to $v_{4}$, as depicted with the blue arrow, then $v_8, v_9, \ldots$ might never be explored.}\label{fig:P2_special1}
	\end{center}\end{figure}
	
	Therefore, the agent has to go back to the vertex with color green---i.\,e., color~$2$---if there are two colored neighbors with colors $1$ and $2$.
	Intuitively speaking, this ensures that the agent first explores all vertices on the left (right) side of the start vertex if there are two vertices with color $2$ on the left (right) side first.
	
	The second thing we notice in \cref{fig:P2_1} is that the adversary might send the agent to different vertices.
	We argue heuristically why it is reasonable to use a forth color in this case.
	The adversary might send the agent to a vertex with only one colored neighbor, consider for example the situations depicted in \cref{fig:P2_special2}.
	\begin{figure}[htbp]\begin{center}
			\begin{subfigure}{.5\textwidth}
				\begin{tikzpicture}
					\tikzstyle{vertex}=[x=1.2cm, y=1.2cm, draw,circle, inner sep=.0pt, minimum size=0.7cm]
					
					\foreach \name/\x in {3/2, 9/5, 11/6}
					\node[vertex] (G-\name) at (\x,0) {$v_{\name}$};
					
					\foreach \name/\x in {4/2, 12/6}
					\node[vertex] (G-\name) at (\x+0.5,-1) {$v_{\name}$};
					
					\foreach \name/\x in {1/1, 13/7}
					\node[vertex, white] (G-\name) at (\x,0) {};
					
					\foreach \name/\x in {2/1, 
						14/7}
					\node[vertex, white] (G-\name) at (\x+0.5,-1) 
					{};
					\foreach \name/\x in {5/3}
					\node[vertex, fill=red] (G-\name) at (\x,0) {$v_{\name}$};
					
					\foreach \name/\x in {10/5}
					\node[vertex, fill=red] (G-\name) at 
					(\x+0.5,-1) {$v_{\name}$};
					\foreach \name/\x in {7/4}
					\node[vertex, fill=green] (G-\name) at (\x,0) {$v_{\name}$};
					
					\foreach \name/\x in {6/3}
					\node[vertex, fill=green] (G-\name) at 
					(\x+0.5,-1) {$v_{\name}$};
					
					\foreach \name/\x in {8/4}
					\node[vertex, fill=cyan] (G-\name) at (\x+0.5,-1) {$v_{\name}$};
					
					\foreach \from/\to in {
						2/3, 3/4,4/5,5/6,6/7,7/8,8/9,9/10,10/11,11/12,12/13,
						1/3, 3/5, 5/7, 7/9, 9/11, 11/13, 
						2/4, 4/6, 6/8, 8/10, 10/12, 12/14}
					\draw (G-\from) -- (G-\to);
					\draw[red, thick, -Stealth] (G-5) edge[out=20, in=160] (G-7);
					\draw[red, thick, -Stealth] (G-7) edge[in=-20,out=-160] (G-5);
					\draw[red, thick, -Stealth] (G-8) edge[out=20, in=160] (G-10);
					\draw[red, thick, -Stealth] (G-10) edge[in=-20,out=-160] (G-8);
					\foreach \from/\to in {5/6, 6/7, 7/8, 8/9}
					\draw[red, thick, -Stealth] (G-\from) -- (G-\to);
				\end{tikzpicture}
				\caption{Situation 1}
		\end{subfigure}
	\end{center}
	\begin{center}
		\begin{subfigure}{.5\textwidth}
			\begin{tikzpicture}
				\tikzstyle{vertex}=[x=1.2cm, y=1.2cm, draw,circle, inner sep=.0pt, minimum size=0.7cm]
				
				\foreach \name/\x in {3/2, 11/6}
				\node[vertex] (G-\name) at (\x,0) {$v_{\name}$};
				
				\foreach \name/\x in {4/2}
				\node[vertex] (G-\name) at (\x+0.5,-1) {$v_{\name}$};
				
				\foreach \name/\x in {1/1, 13/7}
				\node[vertex, white] (G-\name) at (\x,0) {};
				
				\foreach \name/\x in {2/1, 
					14/7}
				\node[vertex, white] (G-\name) at (\x+0.5,-1) {};
				\foreach \name/\x in {5/3}
				\node[vertex, fill=red] (G-\name) at (\x,0) {$v_{\name}$};
				
				\foreach \name/\x in {10/5}
				\node[vertex, fill=red] (G-\name) at (\x+0.5,-1) 
				{$v_{\name}$};
				\foreach \name/\x in {7/4}
				\node[vertex, fill=green] (G-\name) at (\x,0) {$v_{\name}$};
				
				\foreach \name/\x in {6/3, 12/6}
				\node[vertex, fill=green] (G-\name) at (\x+0.5,-1) 
				{$v_{\name}$};
				\foreach \name/\x in {9/5}
				\node[vertex, fill=cyan] (G-\name) at (\x,0) {$v_{\name}$};
				
				\foreach \name/\x in {8/4}
				\node[vertex, fill=cyan] (G-\name) at (\x+0.5,-1) {$v_{\name}$};
				
				\foreach \from/\to in {
					2/3, 3/4,4/5,5/6,6/7,7/8,8/9,9/10,10/11,11/12,12/13,
					1/3, 3/5, 5/7, 7/9, 9/11, 11/13, 
					2/4, 4/6, 6/8, 8/10, 10/12, 12/14}
				\draw (G-\from) -- (G-\to);
				\draw[red, thick, -Stealth] (G-5) edge[out=20, in=160] (G-7);
				\draw[red, thick, -Stealth] (G-7) edge[in=-20,out=-160] (G-5);
				\draw[red, thick, -Stealth] (G-10) edge[out=20, in=160] (G-12);
				\draw[red, thick, -Stealth] (G-12) edge[in=-20,out=-160] (G-10);
				\foreach \from/\to in {5/6, 6/7, 7/8, 8/9, 9/10, 10/11}
				\draw[red, thick, -Stealth] (G-\from) -- (G-\to);
			\end{tikzpicture}
			\caption{Situation 2}
	\end{subfigure}
	\caption{Cases where the agent sees all three colors at once. The agent cannot distinguish the two cases.}\label{fig:P2_special2}
\end{center}\end{figure}
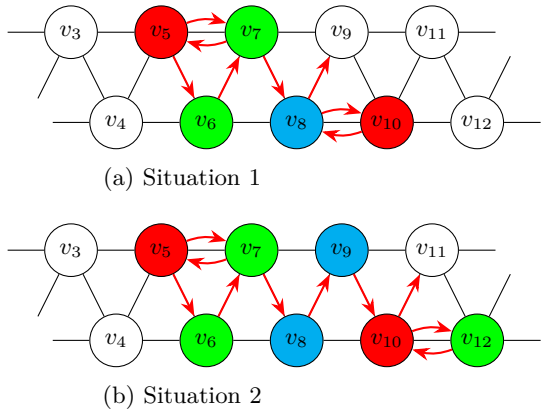

In Situation 1, the adversary sends the agent from $v_8$ to $v_{10}$ instead of $v_9$. 
The agent colors $v_{10}$ red as before and then returns to $v_{8}$.
From $v_{8}$, the agent goes to $v_9$. 
Now, the agent sees all three colors at once.
This is tricky since the agent cannot distinguish this situation from Situation 2 in \cref{fig:P2_special2}, but the color which is farthest from the start vertex is different and the sense of direction might become lost if the agent assigns the wrong color.

In Situation 2, the adversary sends the agent from $v_{10}$ to $v_{12}$ instead of $v_{11}$ and the agent then colors $v_{12}$ green, goes back to $v_{10}$ and from there to $v_{11}$.
As in Situation 1, the input for the algorithm is $c(v)=0$ for the color of the current vertex and $\{1,2,3,0\}$ for the colors of the neighboring vertices.
Hence, on $v_9$ and given as input all three colors red, green, and blue, the agent is wise not to assign any of the three normal colors.
Instead, the agent is going to assign a new color---say gray---and then goes to the only unvisited neighbor, which is $v_{11}$.
This way, the algorithm ensures that there is always a clear direction away from and towards the start vertex.
Thus, the actual coloring pattern might not look like the one in \cref{fig:P2_1} since it depends on the decisions of the adversary.
For example, the coloring pattern might even look like the one in \cref{fig:P2_2}.
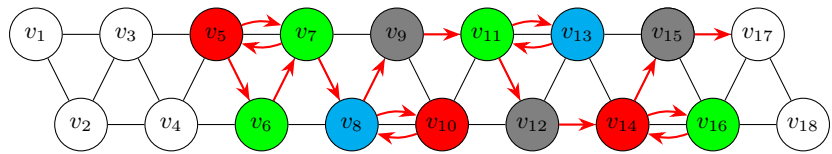
\begin{figure}[htbp]\begin{center}
	\begin{tikzpicture}
		\tikzstyle{vertex}=[x=1.2cm, y=1.2cm, draw,circle, inner sep=.0pt, minimum size=0.7cm]
		
		\foreach \name/\x in {1/1, 3/2, 
			15/8, 17/9}
		\node[vertex] (G-\name) at (\x,0) {$v_{\name}$};
		
		\foreach \name/\x in {2/1, 4/2, 
			18/9}
		\node[vertex] (G-\name) at (\x+0.5,-1) {$v_{\name}$};
		\foreach \name/\x in {5/3}
		\node[vertex, fill=red] (G-\name) at (\x,0) {$v_{\name}$};
		
		\foreach \name/\x in {10/5, 14/7}
		\node[vertex, fill=red] (G-\name) at (\x+0.5,-1) 
		{$v_{\name}$};
		\foreach \name/\x in {7/4, 11/6}
		\node[vertex, fill=green] (G-\name) at (\x,0) {$v_{\name}$};
		
		\foreach \name/\x in {6/3, 16/8}
		\node[vertex, fill=green] (G-\name) at (\x+0.5,-1) 
		{$v_{\name}$};
		\foreach \name/\x in {13/7}
		\node[vertex, fill=cyan] (G-\name) at (\x,0) {$v_{\name}$};
		
		\foreach \name/\x in {8/4}
		\node[vertex, fill=cyan] (G-\name) at (\x+0.5,-1) {$v_{\name}$};
		
		\foreach \name/\x in {9/5, 15/8}
		\node[vertex, fill=gray] (G-\name) at (\x,0) {$v_{\name}$};
		
		\foreach \name/\x in {12/6}
		\node[vertex, fill=gray] (G-\name) at (\x+0.5,-1) {$v_{\name}$};
		
		\foreach \from/\to in {
			1/2, 2/3,3/4,4/5,5/6,6/7,7/8,8/9,9/10,10/11,11/12,12/13,13/14,14/15,15/16,16/17,17/18,
			1/3, 3/5, 5/7, 7/9, 9/11, 11/13, 13/15, 15/17,
			2/4, 4/6, 6/8, 8/10, 10/12, 12/14, 14/16, 16/18}
		\draw (G-\from) -- (G-\to);
		\draw[red, thick, -Stealth] (G-5) edge[out=20, in=160] (G-7);
		\draw[red, thick, -Stealth] (G-7) edge[in=-20,out=-160] (G-5);
		\draw[red, thick, -Stealth] (G-8) edge[out=20, in=160] (G-10);
		\draw[red, thick, -Stealth] (G-10) edge[in=-20,out=-160] (G-8);
		\draw[red, thick, -Stealth] (G-11) edge[out=20, in=160] (G-13);
		\draw[red, thick, -Stealth] (G-13) edge[in=-20,out=-160] (G-11);
		\draw[red, thick, -Stealth] (G-14) edge[out=20, in=160] (G-16);
		\draw[red, thick, -Stealth] (G-16) edge[in=-20,out=-160] (G-14);
		\foreach \from/\to in {5/6, 6/7, 7/8, 8/9, 9/11, 11/12, 12/14, 14/15, 15/17}
		\draw[red, thick, -Stealth] (G-\from) -- (G-\to);
	\end{tikzpicture}
	\caption{The structure with many vertices with color $4$.}\label{fig:P2_2}
\end{center}\end{figure}

We have now covered the situation of the start vertex, where all neighbors are uncolored; the normal cases---namely exactly $1$ colored neighbor, exactly $2$ colored neighbors, and exactly $3$ colored neighbors; and the special case at the beginning, depicted in \cref{fig:P2_special1}. 
There is one case left, namely when the agent is on a vertex of degree $2$.
This is easily handled, the only difference is that on such a vertex, the agent has to start going back towards the start vertex.
The formal description of algorithm 
\textscalt{SquarePathExploration} is given in 
Algorithm~\ref{alg:spe_part1}.

\begin{algorithm}[htbp]
\caption{{}\textscalt{SquarePathExploration}}
\label{alg:spe_part1}
\textbf{Input:} An undirected $P^2$ graph in an online fashion. In each step, the input is $c(v) \in \N$, the color of the current vertex $v$, and $c(v_1), \ldots, c(v_k)$, the colors of $v$'s neighbors, where $k$ is the degree of $v$.\\
\textbf{Output:} In each step, the algorithm outputs $c(v)$, the 
color the agent assigns to $v$, and $i \in [k]$, the vertex the 
agent has to take next.\\
\begin{algorithmic}[1]
	\If{$c(v)=0$}
	\If{$0$ neighbors are colored} \label{line:start_vertex_beginning}
	\LineComment{Situation of start vertex}
	\State $c(v)\coloneqq 1$ \Comment{Assign color $1$ and then \dots}
	\State go to the uncolored neighbor that is presented first \Comment{\dots go exploring.} \label{line:start_vertex_end}
	
	\ElsIf{$1$ neighbor $w$ is colored} \label{line:only_one_colored_neighbor_beginning}
	\LineComment{Normal case $1$: Only $1$ neighbor is colored}
	\State $c(v) \coloneqq (c(w)+1) \modi 3$ \Comment{Assign the next color \dots}
	\State go to $w$ \Comment{\dots and go back to this neighbor.} \label{line:gobackifyouonlyseeoneneighbor} \label{line:only_one_colored_neighbor_end}
	
	\ElsIf{the neighbors have colors $\{1,2,0,0\}$} \label{line:special_case_at_start_beginning}
	\LineComment{Special case to ensure a correct beginning}
	\State $c(v)\coloneqq 2$ \label{line:colorwithtwoneighbors0} \Comment{Assign color $2$\dots}
	\State go to the neighbor with color $2$ \Comment{\dots and go to $w$.} \label{line:special_case_at_start_end}
	
	\ElsIf{$2$ neighbors $w, w'$ are colored} \label{line:conditiontwocoloredneighborsdegree2} \label{line:vertex_of_degree_2_beginning} \label{line:conditiontwocoloredneighborsdegreelarger2} \label{line:two_colored_neighbors_beginning}
	\LineComment{Normal case $2$: $2$ neighbors are colored}
	\LineComment{Assign color and go forward, be careful not to create}
	\LineComment{a triangle with colors $1$,$2$,$3$}	
	\State $c(v) \coloneqq 
	\begin{cases}
		1 \text{, if } \{c(w), c(w')\} \in \{ \{1,3\}, \{3,3\}, \{3,4\}\}\\
		2 \text{, if } \{c(w), c(w')\} \in \{ \{1,1\}, \{1,2\},\{1,4\} \}\\
		3 \text{, if } \{c(w), c(w')\} \in \{ \{2,2\}, \{2,3\},\{2,4\} \}
	\end{cases}$ \label{line:colorwithtwoneighbors1} \label{line:colorwithtwoneighbors2}
	\If{there is an uncolored neighbor $w$} \label{line:cond_gonext_withtwoneighbors1_beginning}
	\State go to $w$
	\Else
	\LineComment{Special case: $v$ has degree $2$ and both neighbors are colored}
	\State go to a neighbor $w$ with $c(w) = (c(v)-1) \modi 3$ 
	\label{line:cond_gonext_withtwoneighbors1_end}
	\EndIf \label{line:two_colored_neighbors_end}
	\ElsIf{$3$ neighbors $w, w', w''$ are colored}\label{line:conditionthreecoloredneighbors} \label{line:three_colored_neighbors_beginning}
	\LineComment{Normal case $3$: $3$ neighbors are colored}
	\State $c(v) \coloneqq 
	\begin{cases}
		4 \text{, if } \{c(w), c(w'), c(w'')\} = \{1,2,3\} \\
		1 \text{, if } \{c(w), c(w'), c(w'')\} \in \{ \{1,1,3\},\{1,3,3\} \}\\
		2 \text{, if } \{c(w), c(w'), c(w'')\} \in \{ \{1,1,2\},\{1,2,2\},\\
		\hspace{4.72cm} \{1,2,4\},\{1,1,4\},\{1,4,4\} \}\\
		3 \text{, if } \{c(w), c(w'), c(w'')\} \in \{ \{2,2,3\},\{2,2,4\},\{2,3,4\},\{2,4,4\} \}
	\end{cases}$ \label{line:colorwiththreeneighbors}
	\If{there is an uncolored neighbor} \label{line:cond_gonext_withthreeneighbors_beginning}
	\State go to some uncolored neighbor
	\Else
	\State go to $z$ with $c(z)=(c(v)-1) \modi 3$ \label{line:cond_gonext_withthreeneighbors_end}
	\EndIf
	\EndIf \label{line:three_colored_neighbors_end}
	\algstore{squarepathalg}
\end{algorithmic}
\end{algorithm}
\begin{algorithm}[htbp]
\label{alg:spe_part2}
\begin{algorithmic}[1]
	\algrestore{squarepathalg}
	\Else \Comment{We have $c(v)>0$}
	\LineComment{On a colored vertex, go forward if there are uncolored neighbors and}
	\LineComment{go back towards the start vertex if there are no uncolored neighbors}
	\If{there is an uncolored neighbor}
	\State go to some uncolored neighbor \label{line:gotonextuncolored}
	\ElsIf{there is a neighbor $z$ with $c(z) = (c(v)-1) \modi 3$}
	\State go to $z$ \label{line:squareexplgobackwards}
	\Else
	\State \textbf{terminate}
	\EndIf
	\EndIf
\end{algorithmic}
\end{algorithm}

We now turn towards the proof that 
\textscalt{SquarePathExploration} successfully explores all square 
paths.
We start by proving that \textscalt{SquarePathExploration} is well 
defined.
\begin{lemma}
\textscalt{SquarePathExploration} is well defined.
\end{lemma}
\begin{proof}
This is clearly true for most of the instructions, we only have to be careful with the color assigning in line~\ref{line:colorwithtwoneighbors2} and the conditions in lines~\ref{line:cond_gonext_withtwoneighbors1_beginning} to \ref{line:cond_gonext_withtwoneighbors1_end} and lines \ref{line:cond_gonext_withthreeneighbors_beginning} to \ref{line:cond_gonext_withthreeneighbors_end}. We check the color assigning in line~\ref{line:colorwithtwoneighbors2} and the conditions in lines \ref{line:cond_gonext_withtwoneighbors1_beginning} to \ref{line:cond_gonext_withtwoneighbors1_end}; the argumentation for lines \ref{line:cond_gonext_withthreeneighbors_beginning} to \ref{line:cond_gonext_withthreeneighbors_end} is similar.
There, we use the evident fact that there can never be three neighbors of the same color.

The color assigning does not cover the case that the two colored neighbors are both colored with color $4$. However, clearly, color $4$ is only assigned if there are three colored neighbors of colors $1$, $2$, and $3$; therefore, there are never two neighbors both with color $4$.
Given this color assignment, the conditions in lines~\ref{line:cond_gonext_withtwoneighbors1_beginning} to \ref{line:cond_gonext_withtwoneighbors1_end} are clearly fulfilled since color $1$ is only assigned if one of the neighbors has color $3$ and similar for the other colors. 
\end{proof}

Next, we show that \textscalt{SquarePathExploration} never colors a 
triangle with the colors $1$, $2$, and $3$.
\begin{lemma}\label{lemma:notrianglesinsquareexpl}
\textscalt{SquarePathExploration} never colors the three vertices 
of a triangle with the colors $1,2$ and $3$, i.e., there are never 
three pairwise adjacent vertices $x,y,z$ such that 
$\{c(x),c(y),c(z)\}=\{1,2,3\}$.
\end{lemma}
\begin{proof}
Consider the step in which \textscalt{SquarePathExploration} colors 
the third vertex $z$ of a triangle, i.e., the other two vertices 
$x$ and $y$ have been colored in earlier steps.
Since at least two neighbors are colored, 
\textscalt{SquarePathExploration} applies one of the 
lines~\ref{line:colorwithtwoneighbors0},~\ref{line:colorwithtwoneighbors1},
 or~\ref{line:colorwiththreeneighbors} to color $z$. 
Clearly, it is never the case that 
\textscalt{SquarePathExploration} assigns color $1$ and the colors 
$2$ and $3$ are present in the neighborhood.
Analogously, \textscalt{SquarePathExploration} never assigns color 
$2$ if colors $1$ and $3$ are present or color $3$ if colors $1$ 
and $2$ are present.
\end{proof}

We are now prepared to prove that the algorithm always terminates.
\begin{lemma}\label{lem:squareexplterminates}
\textscalt{SquarePathExploration} terminates.
\end{lemma}
\begin{proof}
The statement is equivalent to saying that there are no endless 
loops, i.e., there is no sequence of vertices $v_1,v_2,\ldots,v_k, 
v_{k+1}=v_1$ such that \textscalt{SquarePathExploration} always 
goes from $v_i$ to $v_{(i+1) \modi k}$. 
Assume towards contradiction that such a sequence exists. 

If \textscalt{SquarePathExploration} arrives at a vertex that has 
been colored before, it goes to an uncolored neighbor regardless of 
the color of any other vertices.
Therefore, there will be at some point no more uncolored neighbors of the vertices in the loop. 
In other words, the only line that can be applied is line~\ref{line:squareexplgobackwards}. 
Therefore, by assuming without loss of generality that $c(v_1)=1$, the colors in the sequence must be $(c(v_1), c(v_2), c(v_3),c(v_4)\ldots, c(v_k))= (1,2,3,1,\ldots,3)$.

Consider the graph from \cref{fig:P2} and consider the rightmost vertex $v$ that is in the endless loop. 
Say this vertex is $v_i$. 
This vertex has two neighbors $v_{i-1}$ and $v_{i+1}$ that are in the endless loop as well. 
Since we have $(c(v_{i-1}), c(v_i), c(v_{i+1})) \in 
\{(1,2,3),(2,3,1),(3,1,2)\}$ as shown before, we have $v_{i-1} \neq 
v_{i+1}$.
Moreover, since $v_i$ is the rightmost vertex in the loop, $v_{i-1},v_i, v_{i+1}$ form a triangle. 
However, by \cref{lemma:notrianglesinsquareexpl}, this cannot happen; hence, the assumption was wrong.
\end{proof}

\begin{lemma}\label{lem:squareexplterminatesonstartvertex}
\textscalt{SquarePathExploration} only terminates on the start 
vertex.
\end{lemma}
\begin{proof}
\textscalt{SquarePathExploration} only terminates on a vertex $v$ 
without uncolored neighbors and where no neighbor $z$ has color 
$c(z)= (c(v)-1) \modi 3$.

On the one hand, since the start vertex receives color $1$ and its neighbors color $2$ or $4$, these conditions are true for the start vertex as soon as all its neighbors are colored.

On the other hand, all other vertices are colored according to line~\ref{line:colorwithtwoneighbors0},~\ref{line:colorwithtwoneighbors1}, or~\ref{line:colorwiththreeneighbors}. 
We argue that there is always a neighbor with $c(z)= (c(v)-1) \modi 3$. Indeed, this is clear for all cases with the exception of the case when there are two or three colored neighbors, all with color $4$. 
However, clearly, this can never happen as a $4$ is only assigned when there are three colored neighbors of colors $1,2,3$.
Thus, the claim follows.
\end{proof}

We proved that \textscalt{SquarePathExploration} terminates and 
that it terminates on the start vertex. We have yet to show that 
all vertices are indeed explored.
We first show that \textscalt{SquarePathExploration} does not 
create a \emph{hole}, i.e., there is never an uncolored vertex or a 
connected subgraph of uncolored vertices that is surrounded by 
colored vertices unless the uncolored vertex has degree $2$.

Then, we show that, after some initial steps, all vertices on one \emph{side} of the start vertex $s$ are colored.\footnote{To define the word ``side'', we remove $s$ and cut the edge ``opposite of $s$'', that is, the edge between two neighbors of $s$ with the property that when it is removed together with $s$, the graph is not connected any more. If there are two such edges, for example if $v_2 = s$ in \cref{fig:P2_1}, we cut the edge that minimizes the difference in the amount of vertices of the two components. If no such edge exists, then $s$ has degree $2$ and there is just one side.}
Together with the fact that \textscalt{SquarePathExploration} 
returns to the start vertex, the correctness of 
\textscalt{SquarePathExploration} follows.
\begin{lemma}
Apart from vertices of degree 2 and 3, 
\textscalt{SquarePathExploration} never creates a hole.	
\end{lemma}
\begin{proof}
Consider the step in which \textscalt{SquarePathExploration} would 
create a hole. 
In this step, there is just one colored neighbor, but 
\textscalt{SquarePathExploration} would still go to an uncolored 
neighbor.
\end{proof}

\begin{lemma}\label{lem:sideexploration}
As soon as the first two vertices on one side of the start vertex 
are colored, \textscalt{SquarePathExploration} explores all 
vertices on this side.
\end{lemma}
\begin{proof}
Before we start, we observe that the two first vertices on one side of the start vertex that become colored will be vertices that are sequential in the path. In other words, if we draw the graph as for example in \cref{fig:P2_2}, these first two vertices on one side of the start vertex will not be both ``in the top row'' or both ``in the bottom row.''

Having observed this, we now prove the lemma by induction on the number of vertices on one side. If there are only two vertices on one side, there is nothing to prove.
If there are three vertices on one side, we make use of the special case in lines~\ref{line:special_case_at_start_beginning} to \ref{line:special_case_at_start_end} to ensure that the third vertex is colored as well. 
Now suppose the claim is true for $n-1$ vertices on one side.
Add a vertex $v=v_n$ and the according edges to this side.
Let $a$ and $b$ be the two neighbors of $v$ and let $c$ be the common neighbor different from $v$ of $a$ and $b$ and $d$ the fourth neighbor of $b$.
This situation is depicted in \cref{fig:P2explorationproof}.

\begin{figure}
	\begin{center}
		\begin{tikzpicture}
			\tikzstyle{vertex}=[x=1.2cm, y=1.2cm, draw,circle, inner sep=.0pt, minimum size=0.7cm]
			
			\foreach \name/\x in {1/1, 3/3, 4/4}
			\node[vertex] (G-\name) at (\x,0) {};
			
			\foreach \name/\x in {s/2, c/5,	a/6}
			\node[vertex] (G-\name) at (\x,0) {$\name$};
			
			\foreach \name/\x in {10/1, 11/2, 12/3}
			\node[vertex] (G-\name) at (\x+0.5,-1) {};
			
			\foreach \name/\x in {d/4, b/5, v/6}
			\node[vertex] (G-\name) at (\x+0.5,-1) {$\name$};
			
			\foreach \from/\to in {1/s,s/3,3/4,4/c,4/d, c/a, c/b, 3/12, 1/10, a/v, 10/11, 11/12, 12/d, d/b, b/v,s/10,s/11, 3/11, 4/12, c/d, a/b}
			\draw (G-\from) -- (G-\to);
		\end{tikzpicture}
		\caption{Add a vertex $v$ in distance $n$}\label{fig:P2explorationproof}
	\end{center}
\end{figure}
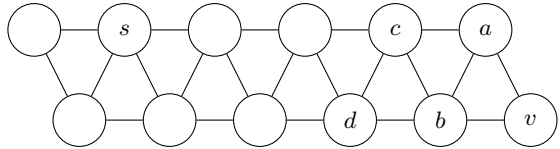

Note that $d$ might be the start vertex but by our assumption, $c$ may not.
By the induction assumption, \textscalt{SquarePathExploration} 
reaches at least $a$ or $b$.

When $a$ is reached first, then $a$ has only one colored neighbor $c$, and the agent returns to $c$.
We may invoke the induction assumption again to ascertain that 
\textscalt{SquarePathExploration} reaches $b$.
At this point, since there are no holes, the only uncolored neighbor is $v$.
Therefore, \textscalt{SquarePathExploration} goes to $v$, colors it 
and then returns to either $a$ or $b$.

When $b$ is reached first, we have two cases. First, 
\textscalt{SquarePathExploration} could go directly to $a$ or $v$.
This can only happen if $c$ and $d$ are colored since a vertex needs two colored neighbors in order to go to an uncolored neighbor.
If the agent visits $a$, it then goes to $v$. If the agent visits 
$v$, it first returns to $b$. Now, $a$ is the only uncolored 
neighbor of $b$ and $b$ has more than one colored neighbor. Hence, 
\textscalt{SquarePathExploration} visits $a$ next and we are done.

Second, \textscalt{SquarePathExploration} could go back to $c$ or 
$d$.
By the induction assumption, going back to $c$ is not possible as this would assume that $d$ is not colored and we have a hole already.
Here, we use the fact that the one of the two colored vertices on one side of the start vertex is in the top row and one is in the bottom row by the initial observation.
Hence, in this case, the agent goes back to $d$ and $c$ is not colored.
On $d$, the agent visits $c$ as the only uncolored neighbor.
The agent colors $c$ and then goes to $a$ as the only uncolored neighbor and again on
$a$, the agent colors $a$ and then goes to $v$ as the only uncolored neighbor.

Thus, \textscalt{SquarePathExploration} explores the larger graph 
without creating a hole.
\end{proof}

The previous lemmas can be combined to show that 
\textscalt{SquarePathExploration} indeed explores all vertices.

\begin{lemma}
\textscalt{SquarePathExploration} explores all vertices.	
\end{lemma}
\begin{proof}
By Lemmas~\ref{lem:squareexplterminates} 
and~\ref{lem:squareexplterminatesonstartvertex}, 
\textscalt{SquarePathExploration} terminates on the start vertex 
$s$.
However, \textscalt{SquarePathExploration} cannot terminate with 
uncolored neighbors.
As soon as two neighbors on one side of $s$ are colored, all vertices of this side are colored by \cref{lem:sideexploration}.
Of course, again by Lemmas~\ref{lem:squareexplterminates} 
and~\ref{lem:squareexplterminatesonstartvertex}, 
\textscalt{SquarePathExploration} then returns to $s$.
If only one neighbor $v$ on one side of the start vertex is colored, then either there is another uncolored vertex $v_2$ on the same side of $s$ that is also a neighbor of $s$ or there is no other vertex on this side of~$s$.
\end{proof}

This finally leads to the desired theorem:
\begin{theorem}
\textscalt{SquarePathExploration} is correct on graphs of the form 
$P^2$ and uses at most $4$ colors.
\end{theorem}

\section{Exploration with Recoloring} \label{sec:recoloring}
In this section, we allow recoloring of already colored vertices.
We prove that in this case, seven colors are enough to explore any graph.
This demonstrates the superior strength of strategies with recoloring.
We present a formal description of our algorithm 
\textscalt{Recolorer} in Algorithm~\ref{alg:recolor}.

Let us describe the basic strategy.
We abuse the term ``color'' and say a vertex either receives label $x$ or a label $1$, $2$ or $3$ together with one of the two colors green and red.
The special label $x$ is assigned to \emph{delete} vertices, that is, to mark vertices to which the algorithm must never return.
The algorithm performs breadth-first search and labels every vertex with $1$, $2$, or $3$ depending on \emph{depth}, i.e., its distance from the start vertex; vertices in depth $k$ receive label $k+1 \modi 3$. Unless a vertex is deleted, the labels are never changed, only the colors.

We imagine the start vertex as top vertex and go from top to bottom when moving further away from the start vertex.
Since vertices in depth $k$ have only neighbors in depth $k-1$, $k$, and $k+1$, this labeling provides a sense of direction for the algorithm.

When \textscalt{Recolorer} is on a vertex $v$ that is in depth $k$, 
we call neighbors in depth $k-1$ \emph{parents} (of $v$), neighbors 
in depth $k+1$ \emph{children} (of $v$), and neighbors in the same 
depth as $v$ \emph{siblings}.
Note that the agent is going to be able to determine whether labeled neighbors are parents, siblings, or children.
To ensure that all vertices in a certain depth are colored, the algorithm colors the vertices from green to red and vice versa, depending on the current phase.

There are green phases and red phases in alternating order.

Initially, the start vertex receives label $1$ and color red.
Then, a green phase starts. In a green phase, the algorithm follows the red labels from $1$ to $2$ to $3$ to $1$ etc. until an unvisited vertex is reached.
This vertex is then marked with the next label and colored green. The algorithm then proceeds to a parent $p$.
If $p$ has unvisited neighbors---these are all children, as we will see---, then one of these neighbors is visited.
Otherwise, if all children of $p$ are colored, the algorithm goes down to a child of the same color as $p$ until it eventually finds an unvisited vertex, which is then colored green before going up to a parent.
As soon as all children of a parent are of another color, the algorithm recolors the parent and then goes up to a grandparent.
This process takes place until the start vertex is reached and all children of the start vertex are colored green.
Now, the start vertex is marked green as well and a red phase starts, which works analogously.

This process would suffice to achieve perpetual exploration.
To ensure that the algorithm terminates, the algorithm deletes a vertex whenever all its children are deleted or if there are no children at all.
An example of an exploration is depicted in \cref{fig:recolorer}.
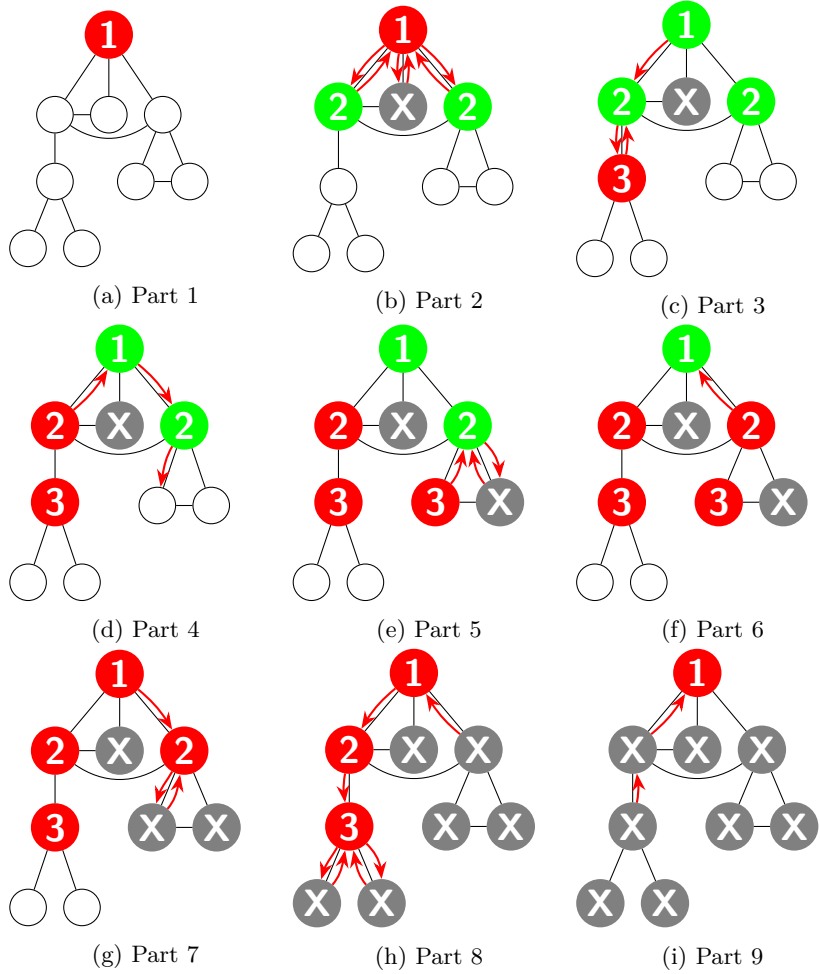
\begin{figure}[htbp]\begin{center}
	\begin{subfigure}{.3\textwidth}
		\begin{forest}
			[1,redded, name=A, for descendants={unvisited}
			[,, name=B1 ,
			[[][]]]
			[,, name=B2]
			[,, name=B3
			[,, name=C1][,, name=C2]]]
			]
			\draw (B1)--(B2) (C1)--(C2);
			\draw (B1) edge[out=-35, in=-145] (B3);
		\end{forest}
		\caption{Part 1}
	\end{subfigure}
	\begin{subfigure}{.3\textwidth}
		\begin{forest}
			[1,redded, name=A, for descendants={unvisited}
			[2,greened, name=B1 ,
			[[][]]]
			[X,deleted, name=B2]
			[2,greened, name=B3
			[,, name=C1][,, name=C2]]]
			]
			\draw (B1)--(B2) (C1)--(C2);
			\draw (B1) edge[out=-35, in=-145] (B3);
			\draw[red, thick, -Stealth] (A) edge[out=-140, in=60] (B1);
			\draw[red, thick, -Stealth] (B1) edge[out=40, in=-120] (A);
			\draw[red, thick, -Stealth] (A) edge[out=-100, in=100] (B2);
			\draw[red, thick, -Stealth] (B2) edge[out=80, in=-80] (A);
			\draw[red, thick, -Stealth] (A) edge[out=-40, in=120] (B3);
			\draw[red, thick, -Stealth] (B3) edge[out=140, in=-60] (A);
		\end{forest}
		\caption{Part 2}
	\end{subfigure}
	\begin{subfigure}{.3\textwidth}
		\begin{forest}
			[1,greened, for descendants={unvisited}
			[2,greened, name=B1 ,
			[3,redded, name=D1 [][]]]
			[X,deleted, name=B2]
			[2,greened, name=B3
			[,, name=C1][,, name=C2]]]
			]
			\draw (B1)--(B2) (C1)--(C2);
			\draw (B1) edge[out=-35, in=-145] (B3);
			\draw[red, thick, -Stealth] (A) edge[out=-140, in=60] (B1);
			\draw[red, thick, -Stealth] (B1) edge[out=-100, in=100] (D1);
			\draw[red, thick, -Stealth] (D1) edge[out=80, in=-80] (B1);
		\end{forest}
		\caption{Part 3}
	\end{subfigure}
	
	\begin{subfigure}{.3\textwidth}
		\begin{forest}
			[1,greened, for descendants={unvisited}
			[2,redded, name=B1 ,
			[3,redded [][]]]
			[X,deleted, name=B2]
			[2,greened, name=B3
			[,, name=C1][,, name=C2]]]
			]
			\draw (B1)--(B2) (C1)--(C2);
			\draw (B1) edge[out=-35, in=-145] (B3);
			\draw[red, thick, -Stealth] (B1) edge[out=40, in=-120] (A);
			\draw[red, thick, -Stealth] (A) edge[out=-40, in=120] (B3);
			\draw[red, thick, -Stealth] (B3) edge[out=-120, in=80] (C1);
		\end{forest}
		\caption{Part 4}
	\end{subfigure}
	\begin{subfigure}{.3\textwidth}
		\begin{forest}
			[1,greened, for descendants={unvisited}
			[2,redded, name=B1 ,
			[3,redded [][]]]
			[X,deleted, name=B2]
			[2,greened, name=B3
			[3,redded, name=C1][X,deleted, name=C2]]]
			]
			\draw (B1)--(B2) (C1)--(C2);
			\draw (B1) edge[out=-35, in=-145] (B3);
			\draw[red, thick, -Stealth] (C1) edge[out=50, in=-100] (B3);
			\draw[red, thick, -Stealth] (B3) edge[out=-45, in=90] (C2);
			\draw[red, thick, -Stealth] (C2) edge[out=130, in=-80] (B3);
		\end{forest}
		\caption{Part 5}
	\end{subfigure}
	\begin{subfigure}{.3\textwidth}
		\begin{forest}
			[1,greened, for descendants={unvisited}
			[2,redded, name=B1 ,
			[3,redded [][]]]
			[X,deleted, name=B2]
			[2,redded, name=B3
			[3,redded, name=C1][X,deleted, name=C2]]]
			]
			\draw (B1)--(B2) (C1)--(C2);
			\draw (B1) edge[out=-35, in=-145] (B3);
			\draw[red, thick, -Stealth] (B3) edge[out=140, in=-60] (A);
		\end{forest}
		\caption{Part 6}
	\end{subfigure}
	
	\begin{subfigure}{.3\textwidth}
		\begin{forest}
			[1,redded, for descendants={unvisited}
			[2,redded, name=B1 ,
			[3,redded [][]]]
			[X,deleted, name=B2]
			[2,redded, name=B3
			[X,deleted, name=C1][X,deleted, name=C2]]]
			]
			\draw (B1)--(B2) (C1)--(C2);
			\draw (B1) edge[out=-35, in=-145] (B3);
			\draw[red, thick, -Stealth] (A) edge[out=-40, in=120] (B3);
			\draw[red, thick, -Stealth] (B3) edge[out=-120, in=80] (C1);
			\draw[red, thick, -Stealth] (C1) edge[out=50, in=-100] (B3);
		\end{forest}
		\caption{Part 7}
	\end{subfigure}
	\begin{subfigure}{.3\textwidth}
		\begin{forest}
			[1,redded, for descendants={unvisited}
			[2,redded, name=B1 ,
			[3,redded [X,deleted, name=E1][X,deleted, name=E2]]]
			[X,deleted, name=B2]
			[X,deleted, name=B3
			[X,deleted, name=C1][X,deleted, name=C2]]]
			]
			\draw (B1)--(B2) (C1)--(C2);
			\draw (B1) edge[out=-35, in=-145] (B3);
			\draw[red, thick, -Stealth] (B3) edge[out=140, in=-60] (A);
			\draw[red, thick, -Stealth] (A) edge[out=-140, in=60] (B1);
			\draw[red, thick, -Stealth] (B1) edge[out=-100, in=100] (D1);
			\draw[red, thick, -Stealth] (D1) edge[out=-120, in=80] (E1);
			\draw[red, thick, -Stealth] (E1) edge[out=50, in=-100] (D1);
			\draw[red, thick, -Stealth] (D1) edge[out=-45, in=90] (E2);
			\draw[red, thick, -Stealth] (E2) edge[out=130, in=-80] (D1);
		\end{forest}
		\caption{Part 8}
	\end{subfigure}
	\begin{subfigure}{.3\textwidth}
		\begin{forest}
			[1,redded, for descendants={unvisited}
			[X,deleted, name=B1 ,
			[X,deleted [X,deleted][X,deleted]]]
			[X,deleted, name=B2]
			[X,deleted, name=B3
			[X,deleted, name=C1][X,deleted, name=C2]]]
			]
			\draw (B1)--(B2) (C1)--(C2);
			\draw (B1) edge[out=-35, in=-145] (B3);
			\draw[red, thick, -Stealth] (D1) edge[out=80, in=-80] (B1);
			\draw[red, thick, -Stealth] (B1) edge[out=40, in=-120] (A);
		\end{forest}
		\caption{Part 9}
	\end{subfigure}
	\caption{An example of an exploration by 
	\textscalt{Recolorer}.}\label{fig:recolorer}
\end{center}\end{figure}

\begin{algorithm}[htbp]
\caption{{}\textscalt{Recolorer}} \label{alg:recolor}
\textbf{Input:} An undirected graph in an online fashion. In each step, the input is $l(v), c(v) \in \N$, which are label and color of the current vertex $v$, and $l(v_1), c(v_1), \ldots, l(v_k), c(v_k)$, the labels and colors of the neighbors of $v$, where $k$ is the degree of $v$.\\
\textbf{Output:} In each step, the algorithm outputs $l(v)$, the label the agent assigns to $v$; $c(v)$, the color the agent assigns to $v$; and $i \in [k]$, the vertex the agent has to take next. It is allowed to change the color of a vertex that has been colored before.\\
\textbf{Description:} For each phase $k$, visit all vertices in depth $k$ and then return to the start vertex. On return, color all vertices green if $k$ is odd, and red otherwise.
\begin{algorithmic}[1]
	\If{all children of $v$ are deleted} \label{line:childrendeleted}
	\LineComment{the entire subtree of the current vertex has been explored}
	\LineComment{or there is nothing to explore}
	\State $l(v) \coloneqq x$ \Comment{$v$ is deleted as well}
	\State go to any parent and if there is no parent, \textbf{terminate.} \label{line:deletevertexandmoveup}
	\ElsIf{$l(v)=0$} \Comment{$v$ has not been visited before}\label{line:uncoloredvertex}
	\State $l(v)\coloneqq \left(\min\limits_{v' \in N(v)} 		l(v')\right)+1 \modi 3$ \label{line:label1}
	\State $c(v)\coloneqq
	\begin{cases}
		\text{green,}&\text{ if there is parent with color red}\\
		\text{red,}&\text{ otherwise}\\
	\end{cases}$ \label{line:color1}
	\State go to any parent and if there is no parent, go to any neighbor \Comment{there is no parent in the case of the start vertex} \label{line:gotoanyneighbor} \label{line:gotoanyparent}
	\Else \Comment{$v$ has been visited before}
	\If{there is an unlabeled child $v'$} \label{line:condunlabeledchild}
	\LineComment{agent is on level $i$ and level $i+1$ has not been completely explored}
	\State go to $v'$ \label{line:gotounlabeledchild}
	\ElsIf{there is a child $v'$ of the same color} \label{line:childofsamecolor}
	\LineComment{agent is on level $i$ and level $j$ with $j>i+1$ has }
	\LineComment{not been completely explored}
	\State	go to $v'$ \label{line:gotochildofsamecolor}
	\ElsIf{there is a non-deleted child and all non-deleted children are of complementary color $k$ ($=\text{green/red}$)}
	\LineComment{the subtree of the current vertex was explored up to some level $i$ and}
	\LineComment{now the agent has to move towards the start vertex}
	\State $c(v) \coloneqq k$ \Comment{the vertex agrees with the color of its children}
	\State go to any parent and if there is no parent, go to any non-deleted child.\label{line:gotoanyneighborfromstartvertex} \label{line:gotoanyparentfromcoloredvertex}
	\EndIf
	\EndIf
\end{algorithmic}
\end{algorithm}

To prove that \textscalt{Recolorer} successfully explores all 
graphs, we first define the term \emph{phase}.
Initially, \textscalt{Recolorer} is in phase 0. 
For each $k\in \naturalnumberpositive$, phase $k$ ends and phase $k+1$ begins immediately after the start vertex receives a color and the agent decides to terminate or to go to a neighbor.
Before we turn towards the main lemma, we start with two general observations.
\begin{observation}\label{obs:recolorer1}
Upon visiting, a vertex becomes labeled.
\end{observation}
\begin{observation}\label{obs:recolorer2}
The label of a vertex $v$ is only changed if $v$ gets deleted.
\end{observation}

\begin{lemma}\label{lem:recolorer}
At the end of phase $k$, all vertices in depth $j$ with $j\le k$ are either deleted or labeled with $(j+1) \modi 3$; all non-deleted and labeled vertices are colored with the same color; and no vertices in depth $k+1$ or lower are visited. Moreover, the agent is not caught in an endless loop in phase $k$.
\end{lemma}
\begin{proof}
We prove this by induction.
Clearly, the start vertex is labeled and colored in phase $0$ 
according to lines~\ref{line:uncoloredvertex} to 
\ref{line:gotoanyneighbor} and the next phase starts.

We assume that the first part of the statement is true for $k-1$: At the end of phase $k-1$ and thus at the beginning of phase $k$,\footnote{Unless phase $k-1$ is the final phase, of course.} the agent has labeled every vertex up to depth $k-1$;
all vertices in depth $j\le k-1$ are either deleted or labeled with label $(j+1) \modi 3$;
all non-deleted vertices up to depth $k-1$ have the same color;
and no vertex in depth $k$ is visited.
Note that we do not need to assume that phase $k-1$ ends since this is implied by the ``at the end of phase $k-1$''.
We are going to show later in the proof that the agent is not caught in an endless loop.

We start with some observations.
\begin{observation}\label{obs:recolorer3}
	\Cref{obs:recolorer2} implies together with the induction assumption that, during phase $k$, on every vertex $v$, the agent can decide whether $v$'s non-deleted neighbors are parents or siblings or children.\footnote{Neighbors with label $(l(v)-1) \modi 3$ are parents, neighbors with label $l(v)$ are siblings and neighbors with label $(l(v)+1) \modi 3$ are children.}
	The only exception are vertices in depth $k$ from the start vertex.
	Here, the agent can decide which non-deleted vertices are parents, but it is not possible to distinguish siblings from children.
\end{observation}

\begin{observation}\label{obs:recolorer6}
	A vertex $v$ only becomes deleted if all its children are deleted (lines~\ref{line:childrendeleted} to \ref{line:deletevertexandmoveup}).
	Deleted children are never visited again. Hence, a deleted vertex is never visited again and we can safely ignore deleted vertices.
\end{observation}

With these observations, we are able to prove a first claim towards proving the first point of the lemma.
\begin{claim}\label{claim:colorofverticesindepthk}
	If a vertex in depth $k$ is visited, it is labeled with $(k+1) \modi 3$ or deleted.
\end{claim}
\begin{proof}[of \cref{claim:colorofverticesindepthk}]
	Once the agent arrives for the first time at some vertex in depth $k$, all siblings and all children have not been visited and are thus unlabeled by Observations~\ref{obs:recolorer1}, \ref{obs:recolorer2}, and \ref{obs:recolorer3}. All parents are labeled with $k \modi 3$ by the induction hypothesis. Therefore, according to line~\ref{line:label1}, the agent assigns label $(k+1) \modi 3$. Note that this label will not be changed later, according to \cref{obs:recolorer2}.
	
	Later in phase $k$, whenever the agent arrives at a vertex in depth $k$, there may be parents with label $k \modi 3$ and siblings with label $(k+1) \modi 3$ and again, according to line~\ref{line:label1}, the agent assigns label $(k+1) \modi 3$.		
\end{proof}

Next, we argue that a vertex cannot be colored several times in a phase, which yields as a byproduct that no vertices in depth $k+1$ or lower are visited.
\begin{claim} \label{claim:colorisonlychangedonceinaphase}
	A vertex $v$ cannot change its color twice in a phase and vertices in depth $k+1$ cannot be reached.
\end{claim}
\begin{proof}[of \cref{claim:colorisonlychangedonceinaphase}]
	Consider any phase $k$.
	Assume towards contradiction that $v$ is the vertex whose color is changed twice first in phase $k$.
	Assume that $v$ is in depth $k-1$ or higher. The only reason to recolor $v$ twice would be if all of $v$'s children were recolored twice. Thus, $v$ could not be the first vertex whose color is changed twice.
	
	Therefore, $v$ must be in depth $k$. The first time $v$ was colored took place because the agent visited $v$ and $v$ had not been visited before.
	To change the color a second time, the agent would have to access $v$ again, either through one of $v$'s children or through one of $v$'s parents or through one of $v$'s siblings. The agent clearly never goes to the sibling of a vertex directly; therefore, $v$ can only be accessed through a parent or through a child. 
	However, $v$'s children can only be accessed through vertices in depth $k$ and whenever a vertex in depth $k$ is visited, the agent colors it with a different color than its parents and then goes up to any parent $p$ in depth $k-1$. In particular, the agent does not go down to any vertex in depth $k+1$. Hence, the only possibility to visit $v$ a second time is through a parent $p$.
	
	From any such $p$, the only way to go down to $v$ again would be if the agent visited $p$ once $p$ is recolored, that is, once $p$ had received the same color as $v$. This is due to the fact that only children of the same color are visited. However, as soon as $p$ is recolored, the agent goes to a higher vertex.
	Therefore, the only way to recolor $v$ would be to change the color of all vertices on a path between $v$ and the start vertex---including the start vertex---and then going down again, which is a contradiction to the assumption that $v$ is colored twice in a phase.
\end{proof}

At the end of phase $k$, that is, when the start vertex becomes recolored, all non-deleted and labeled vertices are thus colored with the same color, which is the second point of the lemma.
The only thing left to prove is that indeed all vertices in depth $k$ are visited.
We prove this by showing first that the agent cannot be caught in an endless loop in phase $k$.

\begin{claim}\label{claim:noendlessloop}
	The agent is not caught in an endless loop in phase $k$. 
\end{claim}
\begin{proof}[of \cref{claim:noendlessloop}]
	Suppose there were an endless loop.
	Then there exists a sequence of vertices $v_1,\ldots, v_r$ such that, for $i \in [r-1]$, the agent goes infinitely often from $v_i$ to $v_{i+1}$ and from $v_r$ back to $v_1$. 
	Since a vertex is not colored twice in phase $k$, all vertices in this sequence have a fixed color. 
	Moreover, there cannot be any unlabeled neighbors since the agent would go to these according to line~\ref{line:condunlabeledchild} and the sequence would change.
	Hence, the only lines that can be applied are lines~\ref{line:childofsamecolor} to \ref{line:gotochildofsamecolor}.
	However, this implies that, for $i \in [r-1]$, $v_i$ is a parent of $v_{i+1}$ and $v_r$ is a parent of $v_1$.
	This contradicts \cref{obs:recolorer3}.
\end{proof}
Now, since there is no endless loop in phase $k$, phase $k$ terminates, which means that the start vertex is recolored.
This can only be the case if all children and children of children etc. are recolored or colored for the first time.
This in turn means in particular that all vertices in depth $k$ are colored. Thus, the proof of \cref{lem:recolorer} is finished.
\end{proof}

\Cref{lem:recolorer} immediately implies the correctness of 
\textscalt{Recolorer}.
We thus obtain the following theorem.
\begin{theorem}
There is an algorithm that never uses more than $7$ colors and explores every graph with recoloring.
\end{theorem}

\section{Conclusion} \label{sec:conclusion_zeromemory}
We investigated graph exploration by a very limited agent and showed tight bounds for the exploration of trees and of general graphs.
Essential for our upper bounds was the idea that the algorithm needs a way to uniquely determine the direct predecessor of a vertex. 
On a high level, our algorithms try to enumerate the number of situations that can occur locally.
Surprisingly, we discovered that this number is not determined by the degree of a vertex, but by the circumference of the graph.
We showed almost tight lower bounds for graphs of a certain circumference as well.

While exploring the driving parameter for the amount of colors necessary and sufficient, we proved almost tight bounds for the exploration of general bipartite graphs and for non-uniform algorithms, where we had to assume that the size of the graph is known.

For our lower bounds, \cref{lem:cantgoback} and \cref{lem:cantleaveuncolored} form the basis of our argumentation.
These lemmas make it possible to construct specific graphs and know exactly---up to renaming of the colors---how a general algorithm behaves on this graph.
Even for the cases where we could not apply the lemmas directly, as for example with the bipartite graphs or the non-uniform algorithms, we used a similar argumentation tailored to the specific case.
As soon as the graph class is limited in such a way that these lemmas do not apply at all, it is quite difficult to prove lower bounds, which we demonstrated in \cref{sec:P2} with an algorithm successfully exploring all square paths by merely using four colors.

Finally, we studied the case where recoloring vertices is allowed.
We showed that in this case, seven colors are already sufficient to explore all graphs, which stands in stark contrast to the amount of colors needed in general.

There are at least three ways to continue this research. First, it would be interesting to generalize \cref{lem:cantgoback} and \cref{lem:cantleaveuncolored} to certain subgraphs or minors.
This would facilitate the analysis of graph exploration in our model for large graph classes.
Second, a natural extension of our model would be to allow a constant number of memory bits. This would in particular enable a direct comparison to the existing research on oblivious graph exploration with known inports.
Third, it might be interesting to analyze how additional information can help the agent. On the one hand, one could study the classical advice complexity model for similar graph exploration models; on the other hand, one could restrict the advice to vertices. This would allow comparing global advice and local advice, which is much more restricted, but can be accessed exactly when needed.

\bibliography{bib}

\end{document}